\documentclass[10pt,journal]{IEEEtran}

\usepackage{enumerate}
\usepackage{tikz}
\usepackage{tikzscale}
\usepackage{verbatim}

\usepackage{graphicx}
\usepackage{amsfonts}
\usepackage{algorithmic}
\usepackage{algorithm}
\usepackage{amssymb}
\usepackage{bbm}
\usepackage{color}
\usepackage{cite}
\usepackage{amsfonts}
\usepackage{relsize}
\usepackage[cmex10]{amsmath}
\usepackage{graphicx}
\usepackage{array}
\usepackage{amssymb}
\usepackage{bbm}
\usepackage{mdwmath}
\usepackage{mdwtab}
\usepackage{fixltx2e}
\usepackage{url}
\usepackage[center,scriptsize]{caption}
\usepackage[font+=scriptsize]{subcaption}
\newtheorem{theorem}{Theorem}[section]
\newtheorem{property}{Property}[section]
\newtheorem{proof}{Proof}
\newtheorem{lemma}{Lemma}[section]

\usepackage{amsmath,amssymb,mathrsfs}
\usepackage{epsfig,epstopdf,epsf,graphicx,graphics}
\usepackage{url}
\usepackage{color}
\usepackage{mathtools}

\newcommand\numberthis{\addtocounter{equation}{1}\tag{\theequation}}

\newcommand{\mcal}{\mathcal}

\newcommand{\mb}{\mathbf}
\newcommand{\mbb}{\mathbb}

\newcommand{\RN}[1]{%
      \textup{\uppercase\expandafter{\romannumeral#1}}%
  }

\input{widebar_hack.tex}

\allowdisplaybreaks
\IEEEoverridecommandlockouts

\title{Wireless Network Simplification:\\Beyond Diamond Networks}

\author{
        \IEEEauthorblockN{Yahya H. Ezzeldin, Ayan Sengupta, Christina Fragouli \\Department of Electrical Engineering, University of California Los Angeles, USA \\ \{yahya.ezzeldin, ayansg, christina.fragouli\}@ucla.edu
    } }
\begin{document}
\maketitle
\begin{abstract}
We consider an arbitrary layered Gaussian relay network with $L$ layers of $N$ relays each, from which we select subnetworks with $K$ relays per layer. We prove that: (i) For arbitrary $L, N$ and $K = 1$, there always exists a subnetwork that approximately achieves $\frac{2}{(L-1)N + 4}$ $\left(\mbox{resp.}\frac{2}{LN+2}\right)$ of the network capacity for odd $L$ (resp. even $L$), (ii) For $L = 2, N = 3, K = 2$, there always exists a subnetwork that approximately achieves $\frac{1}{2}$ of the network capacity. We also provide example networks where even the best subnetworks achieve exactly these fractions (up to additive gaps). Along the way, we derive some results on MIMO antenna selection and capacity decomposition that may also be of independent interest.
\end{abstract}

\IEEEpeerreviewmaketitle
\section{Introduction} \label{sec:intro}
Network simplification looks at the following problem: given a Gaussian relay network, can one provide tight guarantees on the fraction of network capacity that a subnetwork of a given size can always retain (approximately, within an additive constant), irrespective of the channel configurations in the network? For the Gaussian diamond network, where a source communicates with a destination via a single layer of $N$ non-interfering relays, this question was answered in \cite{NOF_Simplification}, where the authors showed that one can always find a subnetwork comprising $K$ relays out of the available $N$, that (approximately) achieves a fraction $\frac{K}{K+1}$ of the network capacity. However, for layered relay networks with more than two hops (i.e., more than one layer of relays between source and destination), the problem has so far remained open.

In this paper, we first characterize the guarantees achievable with subnetworks comprising exactly one relay from each layer, over arbitrary layered Gaussian networks--in other words, we analyze the performance of \emph{routing}. While there exists an abundance of (low-complexity) algorithms to \emph{find} the best route through a network, to the best of our knowledge, there does not exist a result proving universal performance guarantees for routing with respect to an optimal (capacity-achieving) utilization of the entire network. 

Next, we provide guarantees when we select a subnetwork with two relays per layer, from a network with two layers of three relays each. It turns out that even this case is rather challenging to characterize, as will be evidenced by the proofs.

At the heart of characterizing subnetwork performance, is the problem of analyzing how subsets (in terms of antennas) of a MIMO channel behave with respect to the entire MIMO channel. This is because the (approximate) capacity expression for relay networks is given by a minimum of all \emph{cut values} in the network, where each cut is a sum of layer-wise MIMO terms from nodes in the source-side of the cut to those in the destination-side \cite{ADT2011}. Hence, we needed to come up with new results on MIMO antenna selection and MIMO capacity decomposition, that may also be of independent interest.

At a high level, our proofs of capacity guarantees proceed as follows: we assume that for any arbitrary layered network, \emph{all} subnetworks of the given size requirement achieve less than a predetermined fraction of the capacity. This implies that the value of at least one cut in every subnetwork falls below the above fraction of capacity; we collate all links from these failing cuts in every subnetwork into a set $\Lambda$. The crux of the problem is in subsequently demonstrating that inside the set $\Lambda$, there exists a cut for the entire network that has a value less than the network capacity, thus establishing a fundamental contradiction, and guaranteeing the existence of \emph{at least one} subnetwork that achieves a higher capacity than the predetermined fraction. It is in this step that we need to use the MIMO selection and decomposition results to arrive at (approximate) expressions for full network cuts that are compatible with those from failing subnetworks.

In the cases for which we derive guarantees in this paper, we also demonstrate that these are indeed tight (up to an additive constant), i.e., there exist channel instantiations where even the best subnetwork of the given size only achieves the capacity fraction guaranteed by the proof of existence.


\paragraph*{Related Work}
For the Gaussian diamond network, universal capacity guarantees for $k$-relay subnetworks were provided in \cite{NOF_Simplification}. \cite{Javad_Simplification} extended the work of \cite{NOF_Simplification} for some scenarios of the diamond network with multiple antennas at the nodes. In the realm of scheme-specific performance guarantees (as opposed to guaranteeing capacity fractions), the work of \cite{agnihotri12} proved upper bounds on multiplicative and additive gaps for AF-based relay selection, primarily for diamond networks.

Another thread of related work pertains to algorithm design for finding near-optimal subnetworks. \cite{BSF_Selection} and \cite{Ozgur_Selection} made progress in that direction, by providing low-complexity heuristic algorithms for near-optimal relay selection. 

\paragraph*{Organization}
Section~\ref{sec:model} describes the system model and also provides background on expressions and notation used throughout the paper.
Section~\ref{sec:main_results} contains our main results on subnetwork capacity.
Section~\ref{sec:upperbounds} presents the MIMO selection and decomposition lemmas, which are of key importance in the subsequent proofs.
The proofs of our main results are outlined in Sections~\ref{sec:proof_k1} and \ref{sec:proof_k2_N3_L2}.

\section{Model and Preliminaries} \label{sec:model}
Consider a layered network with $L+2$ layers, indexed from $0$ to $L+1$.
The first and last layers consist of one node each--the \textit{Source} ($S$) and the \textit{Destination} ($D$) respectively.
All intermediate layers consist of $N$ relay nodes.

A node in the network is labeled by the tuple $(l,i)$ which represents the layer ($l$) containing the node and the node index ($i$) within that layer.
Following this notation, $S$ and $D$ are labeled as $(0,1)$ and $(L+1,1)$ respectively.
For convenience, we will refer to these two nodes as $S$ and $D$ wherever needed.

At any time $t$, the received signal $Y_j^{(l+1)}[t]$ at node $(l+1,j)$ is a function of the transmitted signals from nodes in layer $l$, 
\begin{align*}
    Y_j^{(l+1)}[t] = \sum_{i = 1}^{N} h^{(l)}_{ij} X_i^{(l)}[t] + W_{j}^{(l+1)}[t]
\end{align*}
where $X_i^{(l)}$ is the transmitted signal from node $(l,i)$, and the additive white Gaussian noise $W_{j}^{(l+1)} \left(\sim \mcal{CN}(0,1)\right)$ is independent of the inputs, as well as of the noise terms at the other nodes.
The (complex) channel gain between the nodes $(l,i)$ and $(l+1,j)$ is denoted by $h^{(l)}_{ij} \in \mbb{C}$. We assume that the transmitted signals from each network node satisfy an average power constraint $\mathbb{E}[|X_i^{(l)}|^2]\leq 1 \quad \forall (l,i)$.

For a more compact representation of the signal flow through the network, we adopt the following notation: we define $\mcal{M}_l$ to be the set of nodes in layer $l$ and $\mcal{M} = \cup_{l=0}^{L+1} \mcal{M}_l$ denotes the set of all nodes in the network.
We collect the channel coefficients from the nodes in $\mcal{M}_l$ to those in $\mcal{M}_{l+1}$ into $\mb{H}^{(l)} \in \mbb{C}^{N \times N}$, where $h^{(l)}_{ij}$ is the element in the $j$-th row and $i$-th column of $\mb{H}^{(l)}$.
For a subset of nodes $\mb{y}_l \subseteq \mcal{M}_l$ and $\mb{y}_{l+1} \subseteq \mcal{M}_{l+1}$, $\mb{H}^{(l)}(\mb{y}_l,\mb{y}^c_{l+1})$ denotes the submatrix of $\mb{H}^{(l)}$ between nodes in $\mb{y}_l$ and nodes in $\mb{y}^c_{l+1}$, where  $\mb{y}^c_{l+1} = \mcal{M}_{l+1} \backslash \mb{y}_{l+1}$.

We define a cut of the network by $\mcal{Y} \subseteq \mcal{M}$, such that $S \in \mcal{Y}$ and $D \in \mcal{Y}^c$. We can represent $\mcal{Y}$ by subsets of nodes in each layer as $\mcal{Y} = \{ S, \mb{y}_1,\mb{y}_2,\cdots,\mb{y}_L \}$, where $\mb{y}_i \subseteq \mcal{M}_i$ for $ 1 \leq i \leq L$.

For such a network, the exact capacity is not known. However, in \cite{ADT2011} the authors prove that it is within a constant gap\footnote{By constant gap, we refer to terms that are independent of the channel coefficients in the network.} from $\widebar{C}$, which is given by

{\small
 \vspace{-0.1in}
\begin{align}
    \widebar{C} \triangleq \min_\mcal{Y} \sum_{l=0}^{L}\log\text{det} \left( \mb{I} + \mb{H}^{(l)}(\mb{y}_l,\mb{y}^c_{l+1}) \mb{H}^{(l)}(\mb{y}_l,\mb{y}^c_{l+1})^\dagger \hspace{-0.02in}\right)
    \label{eq:cutset_orig}
\end{align}
}
Therefore, in the rest of this paper, we work with the approximate capacity $\widebar{C}$ in place of the network capacity to prove our results.

For the proofs in Sections \ref{sec:proof_k1} and \ref{sec:proof_k2_N3_L2}, we additionally use the following notations for the individual links and MIMO capacities. We label the link from node $(l,i)$ to node $(l+1,j)$ by the tuple $(l,i,j)$ and denote its capacity by:
\[
    R^{(l)}_{ij} = \log\left(1 + \vert h^{(l)}_{ij} \vert^2\right)
\]
The capacity of the MIMO channel from the set of nodes $\mb{u}$ in layer $l$ to the set of nodes $\mb{v}$ in layer $l+1$ is denoted by
\[
    M(l)^{\{\mb{v}\}}_{\{\mb{u}\}} = \log\det\left(\mb{I} + \mb{H}^{(l)}(\mb{u},\mb{v}) \mb{H}^{(l)}(\mb{u},\mb{v})^\dagger \right)
\]

\section{Main Results} \label{sec:main_results}
The main results in this paper are summarized in the following two theorems.

\begin{theorem}
For every layered Gaussian relay network with $L$ relay layers and $N$ relays per layer, there exists a subnetwork with $K=1$ relay per layer such that the capacity\footnote{In a line network, the exact capacity is achievable by a Decode-Forward (DF) scheme; hence we refer to the exact capacity $C_{1}^*$ instead of the approximate capacity $\widebar{C}_1^*$} $C_1^*$ of this subnetwork satisfies:
    \begin{equation}
        \begin{aligned}
            C_1^* \geq \begin{cases}
        \dfrac{2}{(L-1)N + 4} \widebar{C}- G_1, & \quad L\ \text{odd} \\
        \dfrac{2}{LN+2} \widebar{C}- G_1, & \quad L\ \text{even}
            \end{cases}
        \end{aligned}
    \end{equation}
    where $G_1 = 4\log(N)$.
    Further, there exists a class of networks such that:
    \[
        C_1^* \leq \begin{cases}
        \dfrac{2}{(L-1)N + 4} \widebar{C}, & \quad L\ \text{odd} \\
        \dfrac{2}{LN+2} \widebar{C}, & \quad L\ \text{even}
            \end{cases}
    \]
    \label{thm:single_path}
\end{theorem}

\paragraph*{Implication}From Theorem \ref{thm:single_path}, we note that for a diamond network (i.e., $L$ = 1) with approximate capacity $\widebar{C}_{dia}$, we get $C_1^* \geq \frac{1}{2} \widebar{C}_{dia} - 4\log(N)$, which is consistent with the result proved in \cite{NOF_Simplification} (with a slightly different gap). This theorem also highlights a key difference between the diamond network and general layered networks: unlike the diamond network, the capacity guarantee on single-path routes in general layered networks is inversely proportional to the total number of relays in the network. Thus, a routing protocol (aided by a genie) that selects the optimal route in a wireless layered network (to reliably reduce the complexity of communication), may incur severe losses that increase with the number of relays in the network. This is in contrast to the capacity that can be achieved by engaging all relays in the network via physical layer cooperation techniques \cite{ADT2011}.
\begin{theorem}
    For every layered Gaussian relay network with $L=2$ relay layers and $N=3$ relays per layer, there exists a subnetwork with $K=2$ relays per layer such that the approximate capacity $\widebar{C}_2^*$ of this subnetwork satisfies
    \begin{equation}
        \widebar{C}_2^* \geq \frac{1}{2} \widebar{C} - G_2
    \end{equation}
    where $G_2 = 1.5\log(3)$.
    Further, there exists a class of networks such that:
    \begin{align*}
        \widebar{C}_2^* \leq \frac{1}{2} \widebar{C}
    \end{align*}
    \label{thm:k2_out_of_3}
\end{theorem}
Theorem \ref{thm:k2_out_of_3} presents a first step towards the characterization of network simplification for layered networks when we select $K>1$ relays per layer. Differently from single-path subnetworks $(K=1)$, where the subnetwork cuts are individual links, the cuts of larger subnetworks span the entire network across different layers and are therefore harder to analyze, as we show in our proof of this theorem.

\section{MIMO Lemmas and Cut Approximations} \label{sec:upperbounds}
In this section, we present two results on Gaussian MIMO channels with i.i.d inputs.
These allow us to develop a class of tunable upper bounds for $\bar{C}$.

We consider an $M \times N$ Gaussian MIMO channel with an i.i.d input vector $X \in \mbb{C}^{M \times 1}$, defined by:
\[
    Y = \mb{H}X + W
\]
where $\mb{H}$ is the Gaussian channel matrix and $W$ is a vector of i.i.d. Gaussian random variables $w_i \left(\sim \mcal{CN}(0,1)\right)$.
With all transmitters limited by (individual) average power constraints (normalized to unity), the capacity of this MIMO channel $C_{M,N}$ is given by:
\[
        C_{M,N} = \log\det\left( \mb{I} + \mb{H}\mb{H}^\dagger \right)
\]
\begin{lemma}
    \label{lemma:zero_gap_thm}
    For an $M \times N$ Gaussian MIMO channel with i.i.d inputs and capacity $C_{M,N}$, the best $K_t \times K_r$ subchannel has a capacity $C_{K_t,K_r}^*$ such that:
    \begin{align}
        C_{M,N} \leq \tfrac{\min(M,N)}{\min(K_t, K_r)}\ C_{K_t,K_r}^* + G
        \label{zero_gap}
    \end{align}
    where $G = \frac{\min(M,N)}{\min(K_t, K_r)}\log\left({M \choose K_t} {N \choose K_r}\right)$
    is a constant, independent of channel coefficients.
\end{lemma}

\begin{proof}
    The proof proceeds by relating the determinants of principal submatrices of a Hermitian matrix to that of the entire matrix. The details (among others related) can be found in \cite{arxiv_antenna_selection_note}. We can also relax the bound on the antenna selection algorithm proposed by Jiang \textit{et al.} in \cite[Theorem 3.1]{varanasi} to a channel independent bound that leads to a similar conclusion. Here we give an explanation of this relaxation for the case $K_t = M, K_r = k$.

    Without loss of generality, assume $N = \min(M,N)$. Let $\lambda_1 \geq \lambda_2 \geq \dots \geq \lambda_N$, be the eigenvalues of the Hermitian matrix $\mb{A} = \mb{I} + \mb{H}\mb{H}^\dagger$, where $\mb{H} \in \mbb{C}^{N\times M}$.
The result in \cite{varanasi} proves that we can select a $k \times M$ submatrix $\breve{\mb{H}}$ from $\mb{H}$ such that the $k \times k$ Hermitian matrix $\breve{\mb{A}} = \mb{I} + \breve{\mb{H}} \breve{\mb{H}}^\dagger$ has eigenvalues $\breve{\lambda}_i$ that satisfy
\[
    \prod_{i=1}^{k} \breve{\lambda}_i \geq \prod_{i=1}^{k} \lambda_i . \prod_{i=1}^{k} \frac{1}{(M-i+1)(N-i+1)}
\]
We will refer to the second product as $\frac{1}{G_v}$. Since an algorithm (to select $\breve{\mb{H}}$ from $\mb{H}$) can at best be optimal, the submatrix $\widehat{\mb{A}}$ with the largest determinant (among the ones obtained from all possible $k \times M$ submatrices of $\mb{H}$) satisfies:
\begin{align}
    \prod_{i=1}^{k} \hat{\lambda}_i &\geq  \prod_{i=1}^{k} \breve{\lambda}_i  \geq \prod_{i=1}^{k} \lambda_i . \frac{1}{G_v}\label{eq:MIMO_Varanasi_ineq_submatrix}
\end{align}
From the assumed ordering of the eigenvalues of the ($N \times N$) matrix $\mb{A}$, we have:
\begin{equation}
\begin{aligned}
    \prod_{i=1}^{k} \lambda_i  &= \prod_{i=1}^{k} \lambda_i^\frac{N-k}{N} \prod_{i=1}^{k} \lambda_i^\frac{k}{N} \geq  \lambda_k^\frac{k(N-k)}{N} \prod_{i=1}^{k} \lambda_i^\frac{k}{N}\\
    &\geq  \lambda_{k+1}^\frac{k(N-k)}{N} \prod_{i=1}^{k} \lambda_i^\frac{k}{N} \geq  \prod_{i=k+1}^{N}\lambda_{i}^\frac{k}{N} \prod_{i=1}^{k} \lambda_i^\frac{k}{N}= \prod_{i=1}^{N} \lambda_i^\frac{k}{N}\\
\end{aligned}
\label{eq:MIMO_Varanasi_eig_ordering}
\end{equation}
Using \eqref{eq:MIMO_Varanasi_eig_ordering} in \eqref{eq:MIMO_Varanasi_ineq_submatrix} and taking logarithm of both sides, we can conclude that:
\[
    \log\det(\mb{I} + \widehat{\mb{H}}\widehat{\mb{H}}^\dagger) \geq \frac{k}{N} \log\det(\mb{I} + \mb{H}\mb{H}^\dagger) - \log(G_v)
\]
By reordering, we have:
\[
    C_{M,N} \leq \frac{N}{k} C_{M,k} + \frac{N}{k} \log(G_v)
\]
\end{proof}

\begin{lemma}
    Consider an $M \times N$ Gaussian MIMO channel with independent inputs and capacity $C$. Let $C_\mcal{A}$ be the capacity of the subchannel where only a subset of the inputs $X_\mcal{A}$ are active.
    If we denote by $\mcal{T}$, the set of transmitters of this MIMO channel, then for any subset $\mcal{A}$ of the transmitters, we have
    \begin{align}
        C \leq C_{\mcal{A}}  + C_{\mcal{A}^c}
    \end{align}
    where $\mcal{A}^c = \mcal{T}\backslash \mcal{A}$ is the complement of $\mcal{A}$ in $\mcal{T}$. The same relation follows if we partition the receivers instead.
    \label{lemma:mimo_add}
\end{lemma}

\begin{proof}
    Let $X_\mcal{A}$ be a subset of the input vector $X$ that refers to the inputs from transmitters in $\mcal{A} \subseteq \mcal{T}$. We define $Y_\mcal{A}$ as:
    \[
            Y_\mcal{A} = \mb{H}_\mcal{A} X_\mcal{A} + W
    \]
    where $\mb{H}_\mcal{A}$ is the channel submatrix constructed by keeping only columns corresponding to $\mcal{A}$.
    We define $Y_{\mcal{A}^c}$ analogously.

    To prove the lemma, we will make use of the submodularity property of symmetric mutual information.
    Let $\Omega = \{X,Y\}$ be the union of the input and output variables of the MIMO channel.
    For $\Lambda \subseteq \Omega$, we can define a submodular function \cite{submodular_functions}:
    \[
        f(\Lambda) = I(\Lambda ; \Omega\backslash\Lambda)
    \]
    Let $\Lambda_1 = \{X_\mcal{A}\}$ and $\Lambda_2 = \{X_{\mcal{A}^c} \}$ be subsets of $\Omega$.
    We have:
    \[
        \begin{aligned}
            f(\Lambda_1) =&\ I(X_\mcal{A};\Omega\backslash\Lambda_1) =  I(X_\mcal{A};X_{\mcal{A}^c},Y) \\
        =&\ I(X_\mcal{A};X_{\mcal{A}^c}) + I(X_\mcal{A};Y \vert X_{\mcal{A}^c}) \\
            \stackrel{(a)}{=}&\ 0 + I(X_\mcal{A};Y_\mcal{A}) \\
        \end{aligned}
    \]
where $(a)$ follows from the independence of $X_\mcal{A}$ and $X_{\mcal{A}^c}$, and the fact that $Y_\mcal{A} = Y - \mb{H}_{\mcal{A}^c}X_{\mcal{A}^c}$.
Similarly, $f(\Lambda_2) =\ I(X_{\mcal{A}^c}\ ; Y_{\mcal{A}^c})$.
    Exploiting submodularity, we have \cite{submodular_functions}:
   \begin{equation}
       f(\Lambda_1) + f(\Lambda_2) \geq f(\Lambda_1 \cup \Lambda_2) + f(\Lambda_1 \cap \Lambda_2)
       \label{eq:submodular}
   \end{equation}
    From \eqref{eq:submodular}, we have:
    \[
        \begin{aligned}
            I(X_\mcal{A};Y_\mcal{A}) + I(X_{\mcal{A}^c}\ ;Y_{\mcal{A}^c}) &\geq I(X_\mcal{A} , X_{\mcal{A}^c}\ ;Y) + I(\phi;X,Y) \\
            &= I(X;Y)
        \end{aligned}
    \]
    Since this is true for any arbitrary distribution of $X$, by maximizing both sides of the inequality, we get the statement of the lemma.
    Applying the same arguments on two complementary sets of receivers also gives a similar result.\\
\end{proof}

With these lemmas at hand, we can now proceed to develop some useful upper bounds on $\widebar{C}$.
Define the functions $g_1^\mcal{A}, g_2^\mcal{A}$ and $g_3^k$ as:
\[
    \begin{aligned}
        g^\mcal{A}_1(\mb{u},\mb{v},l) &\triangleq  M(l)_{\{\mb{u}_{\mcal{A}}\}}^{\{\mb{v}\}} + M(l)_{\{\mb{u}_{\mcal{A}^c}\}}^{\{\mb{v}\}},\quad \mb{u}_\mcal{A} \subseteq \mb{u}\\
        g^\mcal{A}_2(\mb{u},\mb{v},l) &\triangleq  M(l)_{\{\mb{u}\}}^{\{\mb{v}_\mcal{A}\}} + M(l)_{\{\mb{u}\}}^{\{\mb{v}_{\mcal{A}^c}\}},\quad \mb{v}_\mcal{A} \subseteq \mb{v}\\
        g^k_3(\mb{u},\mb{v},l) &\triangleq \tfrac{\min(\vert \mb{u}\vert,\vert \mb{v} \vert)}{k} \max_{\substack{\mb{u}_k \subset \mb{u},\ \mb{v}_k \subset \mb{v},\\ \vert \mb{u}_k \vert = \vert \mb{v}_k \vert = k}} M(l)_{\{\mb{u}_k\}}^{\{\mb{v}_k\}} \\
    \end{aligned}
\]
Note that the Lemmas \ref{lemma:zero_gap_thm} and \ref{lemma:mimo_add} imply that:
\begin{subequations}
    \begin{align}
        M(l)_{\{\mb{u}\}}^{\{\mb{v}\}} &\leq  g^\mcal{A}_1(\mb{u},\mb{v},l)  \quad \forall \mb{u}_{\mcal{A}} \subseteq \mb{u} \label{eq:upper_functions_a} \\
        M(l)_{\{\mb{u}\}}^{\{\mb{v}\}} &\leq  g^\mcal{A}_2(\mb{u},\mb{v},l)  \quad \forall \mb{v}_{\mcal{A}} \subseteq \mb{v}  \label{eq:upper_functions_a2} \\
        M(l)_{\{\mb{u}\}}^{\{\mb{v}\}} &\leq  g^k_3(\mb{u},\mb{v},l) + G_k \quad \forall 1 \leq k \leq \min(\vert \mb{u} \vert, \vert \mb{v} \vert)  \label{eq:upper_functions_b}
    \end{align}
    \label{eq:upper_functions}
\end{subequations}
where $G_k = \frac{\min(\vert \mb{u}\vert,\vert \mb{v} \vert)}{k} \log\left({\vert \mb{u}\vert \choose k}{ \vert\mb{v}\vert \choose k} \right)$.
\\
\begin{lemma}
    For a layered Gaussian relay network with $L$ layers, $N$ relays per layer, define $f(\mcal{Y}) \triangleq \sum_{l=0}^{L} f_l(\mb{y}_l,\mb{y}_{l+1}^c)$, where $f_l(\mb{y}_l,\mb{y}_{l+1}^c)$ is some ordered application of $g_1^\mcal{A}(\cdot,\cdot,l)$, $g_2^\mcal{A}(\cdot,\cdot,l)$  and $g_3^k(\cdot,\cdot,l)$ ($1\leq k \leq K$) on $M(l)_{\{\mb{y}_l\}}^{\{\mb{y}_{l+1}^c\}}$. Then we have
    \begin{align*}
    \widebar{C} &\leq \min_{\mcal{Y}} \min_{f \in \mcal{F}_K}f(\mcal{Y}) + \widetilde{G} \numberthis \label{eq:lemma_upperbound} \\
    \text{where }\ &\widetilde{G} < (2N + 2N^3 (L-1))\log \left({N \choose \min\left( \frac{N}{2},\ K\right)} \right)
    \end{align*}
    and $\mcal{F}_K$ consists of all possible (valid) layer-wise compositions of $g_1^\mcal{A}, g_2^\mcal{A}$ and $g_3^k$ ($1\leq k \leq K$) in any order\footnote{The constant gap $\widetilde{G}$ is very crude and can be improved. Our purpose however, is to show that applying Lemmas \ref{lemma:zero_gap_thm} and \ref{lemma:mimo_add}, we get upper bounds that are only a constant gap away from $\widebar{C}$.}.
    \label{lemma:cutset_upperbound}
\end{lemma}

\begin{proof}
    Note that only applying \eqref{eq:upper_functions_b} introduces a constant term.
    Therefore, to get a handle on the largest constant arising, we can naively calculate an upper bound on how many times we can apply \eqref{eq:upper_functions_b} and then penalize by the largest possible constant for each time we use \eqref{eq:upper_functions_b}.
   Select an arbitrary cut $\mcal{Y}$.
   The capacity $Cut(\mcal{Y})$ of this cut is:
   \begin{align}
           Cut(\mcal{Y}) =  M(0)_{\{S\}}^{\{\mb{y}_1^c\}} + \sum_{l=1}^{L-1} M(l)_{\{\mb{y}_l\}}^{\{\mb{y}_{l+1}^c\}} + M(L)_{\{\mb{y}_L\}}^{\{D\}}
       \label{eq:one_cut_general_upperbound}
   \end{align}
   For any  $l$, $1\leq l \leq L-1$, each application of \eqref{eq:upper_functions_b} removes at least one link from the channel between $\mb{y}_l$ and $\mb{y}_{l+1}^c$.
   This means that we can apply \eqref{eq:upper_functions_b} at most $N^2$ times.
   For $l=0$ and $l=L$, we can apply \eqref{eq:upper_functions_b} up to $N$ times.
   We can upper bound the constant $G_k$ for each application as
   \begin{align*}
           &\frac{\min(\vert \mb{y}_l \vert, \vert \mb{y}_{l+1} \vert )}{k} \log\left({\vert \mb{y}_l\vert \choose k}\binom{\vert\mb{y}_{l+1}^c \vert}{k} \right)\\
           &\leq \tfrac{N}{k} \log \left( \binom{N}{k}^2\right) \stackrel{(a)}{\leq} 2N \log \left( \binom{N}{\min(K,\left \lfloor ^N/_2 \right\rfloor )}\right)
       \end{align*}
       where $(a)$ follows from a property of binomial coefficients, that for fixed $N$, $\max \binom{N}{i} = \binom{N}{\lfloor ^N/_2 \rfloor}$.
   For $l = 0$ and $l = L$, the same argument follows without the pre-log term (since $\vert \mb{y}_0 \vert = 1$, $\vert \mb{y}_{L+1}^c \vert = 1$). Plugging these bounds into \eqref{eq:one_cut_general_upperbound}, we get
   \[
       \begin{aligned}
           Cut(\mcal{Y}) &\leq  \min_{f \in \mcal{F}_K} \sum_{l=0}^{L} f_l(\mb{y}_l,\mb{y}_{l+1}^c, l) + \widetilde{G} \\
       \widetilde{G} &< (2N + 2N^3(L-1)) \log\left( \binom{N}{\min(K,\left \lfloor ^N/_2 \right\rfloor)}\right)
       \end{aligned}
   \]
   Since the selected cut $\mcal{Y}$ is arbitrary, \eqref{eq:lemma_upperbound} follows directly.\\
\end{proof}
\emph{Remark:} We can take a subset of $\mcal{F}_K$ by considering certain orderings while applying \eqref{eq:upper_functions_a}, \eqref{eq:upper_functions_a2} and \eqref{eq:upper_functions_b}. In such a case, using the same constant $\widetilde{G}$ as above gives a looser upper bound on $\widebar{C}$.
However, adding structure to how we apply \eqref{eq:upper_functions_a}, \eqref{eq:upper_functions_a2} and \eqref{eq:upper_functions_b} gives us a better handle on the constant $\widetilde{G}$, since the maximum incurred constant can be reduced.
We show such examples in the next Sections when we prove Theorems \ref{thm:single_path} and \ref{thm:k2_out_of_3}.

\section{Proof of Theorem \ref{thm:single_path}} \label{sec:proof_k1}
For this proof, we use Lemma \ref{lemma:cutset_upperbound} with $K=1$.
Furthermore, we restrict $\mcal{F}_1$ to contain a single function $f(\mcal{Y})$ where $f_l(\cdot,\cdot,l) = g_3^1(\cdot,\cdot,l)$.
In this case, we have
\[
    \widebar{C}  \leq \min_{\mcal{Y}} \sum_{l=0}^{L} \left[ \min(\vert \mb{y}_l\vert,\vert \mb{y}_{l+1}^c \vert)  \max_{i \in \mb{y}_l , j \in \mb{y}_{l+1}^c } R^{(l)}_{ij} \right] + \widetilde{G}
\]
Revisiting the calculation of  $\widetilde{G}$ in Lemma \ref{lemma:cutset_upperbound} tells us that instead of applying \eqref{eq:upper_functions_b} up to $N^2$ or $N$ times, we now need to apply it only once per layer and as a result, $\widetilde{G} = (2 + 2N(L-1))\log(N)$.
Throughout this section, we use $\widetilde{C}$ as
\begin{equation}
    \widetilde{C} \triangleq \min_{\mcal{Y}} \sum_{l=0}^{L} \left[ \min(\vert \mb{y}_l\vert,\vert \mb{y}_{l+1}^c \vert)  \max_{i \in \mb{y}_l , j \in \mb{y}_{l+1}^c } R^{(l)}_{ij} \right]
    \label{eq:single_path_C_tilde}
\end{equation}
Therefore, we have
\begin{align}
    \widebar{C} \leq \widetilde{C} + (2 + 2N(L-1))\log(N)
    \label{eq:tilde_to_bar_relation_k1}
\end{align}

From \eqref{eq:single_path_C_tilde}, we can prove an interesting property about the structure of $\widetilde{C}$.
\begin{property}
    Define $T(\mcal{Y})$, for a certain partition $\mcal{Y}$ as:
    \begin{equation}
        T(\mcal{Y}) \triangleq \sum_{l=0}^{L} \min(\vert {\bf y}_l\vert,\vert {\bf y}_{l+1}^c \vert)
        \label{eq:additive_terms}
    \end{equation}
    Then we have,
        \begin{align*}
        \max_{\mcal{Y}} T(\mcal{Y}) \leq
        \begin{cases}
        \frac{(L-1)N}{2} + 2, & \quad L\ \text{odd} \\
        \frac{LN}{2} +2 , & \quad L\ \text{even}
            \end{cases}
        \end{align*}
 \label{prpty:minimalcuts_N2}
\end{property}
\begin{proof}
    See Appendix \ref{appendix:T_Y_proof}.
\end{proof}
Using Property \ref{prpty:minimalcuts_N2}, we can now prove Theorem \ref{thm:single_path} by contradiction.

Define $\alpha_o$ and $\alpha_e$ as
\[
    \alpha_o \triangleq  \frac{2}{(L-1)N + 4},\quad \alpha_e \triangleq \frac{2}{LN + 2}
\]
Consider a network with odd number of relay layers $L$.
Let $\alpha_o = \frac{2}{(L-1)N + 4}$ and assume that for all subnetworks with $K=1$ relay per layer, the capacity is less than $\alpha_o \widetilde{C}$. Then, in each such subnetwork, there exists a link $(l,i,j)$ such that $R^{(l)}_{ij}< \alpha_o \widetilde{C}$.
Let $\mcal{B}$ be the set of all such links, i.e., $\mcal{B} = \{(l,i,j): R^{(l)}_{ij} < \alpha_o\widetilde{C}\}$.
Since $\mcal{B}$ collates cuts (which are singleton links for routes) from all subnetworks, $\mcal{B}$ separates the \emph{Source} from the \emph{Destination}. Hence, there exists a $\hat{\mcal{B}} \subseteq \mcal{B}$ such that $(l,i,j) \in \hat{\mcal{B}}$ represent links between $\mcal{Y}_{\hat{\mcal{B}}}$ and $\mcal{Y}_{\hat{\mcal{B}}}^c$, where $\mcal{Y}_{\hat{\mcal{B}}} = \{ S, \mb{\hat{y}}_1,\mb{\hat{y}}_2,\cdots,\mb{\hat{y}}_L \}$.
From $\mcal{Y}_{\hat{\mcal{B}}}$, we have:
    \[
        \begin{aligned}
            \widebar{C} &\stackrel{(a)}\leq \sum_{l=0}^{L} \left[ \min(\vert \mb{\hat{y}}_l\vert,\vert \mb{\hat{y}}_{l+1}^c \vert)  \max_{i \in \mb{\hat{y}}_l , j \in \mb{\hat{y}}_{l+1}^c } R_{ij} \right]\\
            &\stackrel{(b)}{<} \sum_{l=0}^{L} \left[ \min(\vert \mb{\hat{y}}_l\vert,\vert \mb{\hat{y}}_{l+1}^c \vert)\  \alpha_o \widetilde{C} \right]\\
                &= \alpha_o \widetilde{C}\ T(\mcal{Y}_{\hat{B}}) \\
                &\stackrel{(c)}{\leq} \frac{2}{(L-1)N + 4}\ \widetilde{C} \left(\frac{(L-1)N}{2}+2\right) = \widetilde{C}
        \end{aligned}
    \]
    where $(a)$ follows from \eqref{eq:single_path_C_tilde}, $(b)$ follows from the fact that links part of the cut characterized by $\mcal{Y}_{\hat{\mcal{B}}}$ have capacities strictly less than $\alpha_o \widetilde{C}$ and $(c)$ follows from Property \ref{prpty:minimalcuts_N2}.
    This results in the contradiction $\widetilde{C} < \widetilde{C}$.

    Therefore, for any network with odd number of relay layers $L$, there exists a subnetwork with one relay per layer such that:
\[
    \begin{aligned}
    C_1^* \geq& \frac{2}{(L-1)N + 4}\ \widetilde{C} \\
    \stackrel{(a)}{\geq}& \frac{2}{(L-1)N + 4}\ \widebar{C} - \frac{(4 + 4(L-1)N)\log(N)}{(L-1)N + 4}\\
    \geq& \frac{2}{(L-1)N + 4}\ \widebar{C} - 4\log(N)
    \end{aligned}
\]
where $(a)$ is implied by \eqref{eq:tilde_to_bar_relation_k1}.

Using a similar argument for even $L$, a network with all single-path subnetworks having capacity less that $\alpha_e \widetilde{C}$ will result in a cut $\mcal{Y}_{\hat{\mcal{B}}}$ such that
    \[
        \begin{aligned}
            \widebar{C} &\leq \sum_{l=0}^{L} \left[ \min(\vert \mb{\hat{y}}_l\vert,\vert \mb{\hat{y}}_{l+1}^c \vert)  \max_{i \in \mb{\hat{y}}_l , j \in \mb{\hat{y}}_{l+1}^c } R_{ij} \right]\\
            &< \sum_{l=0}^{L} \left[ \min(\vert \mb{\hat{y}}_l\vert,\vert \mb{\hat{y}}_{l+1}^c \vert)\  \alpha_e \widetilde{C} \right] = \alpha_e \widetilde{C}\ T(\mcal{Y}_{\hat{B}}) \\
            &\leq \frac{2}{(L-1)N + 4}\ \widetilde{C} \left(\frac{(L-1)N}{2}+2\right) = \widetilde{C}
        \end{aligned}
    \]
    which again yields a contradiction.
    Therefore, for even $L$,
\begin{align*}
    C_1^* \geq& \frac{2}{LN + 2}\ \widetilde{C} \\
    {\geq}& \frac{2}{LN + 2}\ \widebar{C} - \frac{4 + 4(L-1)N\log(N)}{LN + 2} \\
                    \geq& \frac{2}{LN + 2}\ \widebar{C} - 4\log(N)
\end{align*}
This completes our proof of the lower bound.

To prove that this worst case bound is tight (within a constant gap), it suffices to provide example networks where the maximum capacity of any subnetwork choosing $K=1$ relays per layer is
\begin{align*}
    C_1^* = \alpha_o\ \widebar{C} \quad \text{($L$ odd)}\quad ,\quad  C_1^* = \alpha_e\ \widebar{C} \quad \text{($L$ even)}
\end{align*}

    For odd $L$, consider the example network illustrated in Fig. \ref{fig:converse_odd_layers_relays} for $L=5$ layers of relays.
        The general construction for arbitrary odd $L$ for this network is:
    \[
        \begin{aligned}
            &\ R^{(0)}_{Si} = R^{(L)}_{ND}= \widebar{C}\ \ \quad\quad\quad\quad\quad\quad \forall 1 \leq i \leq N-1 \\
            &\ R^{(0)}_{SN}= R^{(L)}_{1D} = \frac{2}{(L-1)N+4}\ \widebar{C}\\
            &\ R^{(L)}_{iD} = 0\ \ \quad\quad\quad\quad\quad\quad\quad\quad\quad\quad \forall 2 \leq i \leq N-1 \\
            l\ \text{odd} & \quad (l \neq 0, L): \\
            &\ R^{(l)}_{ii} = \frac{2}{(L-1)N + 4}\ \widebar{C} \ \ \quad\quad\quad\quad \forall 1 \leq i \leq N-1 \\
            &\ R^{(l)}_{Ni} = R^{(l)}_{iN} = \widebar{C}  \quad\quad\quad\quad\quad\quad \forall 1 \leq i \leq N-1 \\
            &\ R^{(l)}_{ij} = 0 \quad\quad\quad\quad\quad\quad\quad\quad 1 \leq i,j \leq N-1,\ i \neq j \\
            l\ \text{even} & \quad (l \neq 0, L):\\
            &\ R^{(l)}_{NN} = \frac{2}{(L-1)N+4}\widebar{C} \\
            &\ R^{(l)}_{ij} = \widebar{C} \quad\quad\quad\quad\quad\quad\quad\quad\quad\quad\quad i\neq N \& j \neq N
        \end{aligned}
    \]
It is easy to see that for all cuts except the one highlighted in Fig. \ref{fig:converse_odd_layers_relays}, the capacity is greater than or equal $\widebar{C}$.
Particularly, if any node on the $Source$ side switches to the $Destination$ side (or vice versa), a link of capacity $\widebar{C}$ is added to the cut value.

Any path from $S$ to $D$ in Fig. \ref{fig:converse_odd_layers_relays} has at least one link with capacity $\alpha_o \widebar{C}$ and therefore, all single-path subnetworks have capacity of at most $\frac{2}{(L-1)N +4} \widebar{C}$.

For even $L$, we consider the network illustrated in Fig. \ref{fig:converse_even_layers_relays}, which follows the general construction:
    \[
        \begin{aligned}
            &\ R^{(0)}_{Si} = R^{(L)}_{iD} = \widebar{C}\quad\quad\quad\quad 1 \leq i \leq N-1 \\
            &\ R^{(0)}_{SN} = R^{(L)}_{ND} = \dfrac{2}{LN+2}\widebar{C}\\
            l\ \text{odd} & \quad (l \neq 0, L):\\
            &\ R^{(l)}_{ii} = \dfrac{2}{LN+2}\widebar{C} \quad \quad\quad 1 \leq i \leq N-1 \\
            &\ R^{(l)}_{ij} = 0 \quad\quad\quad\quad \quad\quad 1 \leq i,j \leq N-1,\ i \neq j \\
            l\ \text{even} & \quad (l \neq 0, L):\\
            &\ R^{(l)}_{NN} = \dfrac{2}{LN+2}\widebar{C} \\
            &\ R^{(l)}_{ij} = \widebar{C} \quad \quad\quad\quad\quad\quad i\neq N \& j \neq N, \\
        \end{aligned}
    \]
    Similar to the case for odd $L$, the highlighted cut is the minimum cut, since it avoids all links with capacity $\widebar{C}$.
    Since the Figure illustrates a cut, all paths from $S$ to $D$ include at least one link belonging to the highlighted cut. Therefore any path from $S$ to $D$ has a capacity of at most $\frac{2}{LN+2} \widebar{C}$.

    This concludes our proof of Theorem \ref{thm:single_path}.
    \begin{figure}
        \centering
        \includegraphics[width=0.45\textwidth]{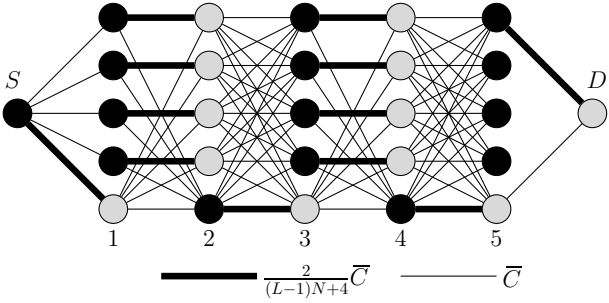}
        \caption{Example network with $N=5$ relays per layer and odd $L=5$ relay layers. Dark nodes represent nodes on the $Source$ side of the cut.}
        \label{fig:converse_odd_layers_relays}
    \end{figure}

    \begin{figure}
        \centering
        \includegraphics[width=0.47\textwidth]{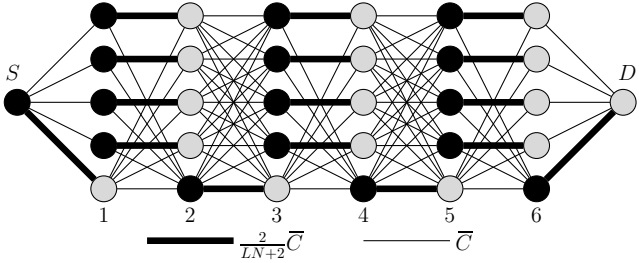}
        \caption{Example network with $N=5$ relays per layer and even $L=6$ relay layers. Dark nodes represent nodes on the $Source$ side of the cut.}
        \label{fig:converse_even_layers_relays}
    \end{figure}

\section{ Proof of Theorem \ref{thm:k2_out_of_3}} \label{sec:proof_k2_N3_L2}

    \begin{figure}
        \centering
        \includegraphics[width=0.3\textwidth]{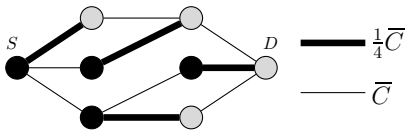}
        \caption{Example network with $L=2$ relay layers and $N=3$ relay per layer. Dark nodes are on the Source side of the cut.}
        \label{fig:converse_L2_N3}
    \end{figure}

In this section, we consider a layered relay network with $N=3$ and $L=2$. First, we provide an example network in Fig.~\ref{fig:converse_L2_N3}, with a minimum cut $\widebar{C}$ (highlighted in dark), in which it is easy to see that for every subnetwork comprising two relays per layer, the subnetwork capacity is at most $\frac{1}{2}\widebar{C}$.

Towards the lower bound in Theorem \ref{thm:k2_out_of_3}, we prove that by enforcing some structure on $\mcal{F}_K$, we can reduce the result of Lemma \ref{lemma:cutset_upperbound} to
\begin{align} \label{eq:bar_to_f_K2_L2}
    \widebar{C} \leq \min_{\mcal{Y}} \min_{f \in \mcal{F}_\mcal{Y}} f(\mcal{Y}) + 3\log(3)
\end{align}
where $\mcal{F}_\mcal{Y}$ considers only certain ways of applying \eqref{eq:upper_functions}, depending on the structure of the cut $\mcal{Y}$.
We list the different realizations of $\mcal{F}_{\mcal{Y}}$ in Table \ref{tab:class_to_f_matching}. The calculation of $f(\mcal{Y})$ and the constant $G_{\mcal{Y}}$ is the topic of Appendix \ref{appendix:gaps}.

To prove \eqref{eq:bar_to_f_K2_L2}, consider the following:
\begin{align*}
\widebar{C} &= \min_{\mcal{Y}} Cut(\mcal{Y}) \\
&\leq \min_{\mcal{Y}} \left\{ \min_{f \in \mcal{F}_{\mcal{Y}}}f(\mcal{Y}) + G_{\mcal{Y}} \right\}\\
&\leq \min_{\mcal{Y}} \left\{ \min_{f \in \mcal{F}_{\mcal{Y}}}f(\mcal{Y}) \right\} + \max_{\mcal{Y}}\ G_{\mcal{Y}}\\
&\stackrel{(a)}{=} \min_{\mcal{Y}} \min_{f \in \mcal{F}_{\mcal{Y}}} f(\mcal{Y}) + 3\log(3)
\end{align*}
where $(a)$ follows from the constants $G_{\mcal{Y}}$ in Table \ref{tab:class_to_f_matching}.

    \begin{table*}
        \caption{Classes of cuts, their upper bounding functions and the constants incurred (Summary of discussion in Appendix \ref{appendix:gaps}).\\The cuts $\mcal{Y}$ and $\mcal{Y}^r$ are reflections of one another and therefore have similar cut structure (reflected).}
        \label{tab:class_to_f_matching}
        \centering
        \bgroup
        \def\arraystretch{1.6}
        \begin{tabular}{|l|l|l|l|}
            \hline
            \multicolumn{2}{|c|}{$\mcal{Y} (\mcal{Y}^r)$} & \multicolumn{1}{c|}{$\displaystyle \mcal{F}_{\mcal{Y}}$ ( $\displaystyle\mcal{F}_{\mcal{Y}^r}$ are reflected accordingly)} & $G_{\mcal{Y}}$\\
            \hline\hline
            (1) & \begin{tabular}{l}$\mcal{Y} = \{S,\phi,\phi\}$\\$\mcal{Y}^r = \{S,\{1,2,3\},\{1,2,3\}\}$\end{tabular} & $ \displaystyle f(\mcal{Y}) = \max_{i \in \{1,2,3\}} R^{(0)}_{Si}$ & $\log(3)$ \\ \hline
            (2) & \begin{tabular}{l}$\mcal{Y} = \{S,\{3\},\phi\}$\\$\mcal{Y}^r = \{S,\{1,2,3\},\{1,2\}\}$\end{tabular} & $ \displaystyle   f(\mcal{Y}) = \max_{i \in \{1,2\}} R^{(0)}_{Si} + \max_{i \in \{1,2,3\}} R^{(1)}_{3i}$ & $\log(2)$ \\ \hline
            (3) & \begin{tabular}{l}$\mcal{Y} = \{S,\{3\},\{1\}\}$\\$\mcal{Y}^r = \{S,\{2,3\},\{1,2\}\}$\end{tabular} & $ \displaystyle   f(\mcal{Y}) = \max_{i \in \{1,2\}} R^{(0)}_{Si} + \max_{i \in \{2,3\}} R^{(1)}_{3i} + R^{(0)}_{1D} $ & $ 2\log(2)$ \\ \hline
            (4) & \begin{tabular}{l}$\mcal{Y} = \{S,\{3\},\{1,2\}\}$\end{tabular} & $ \displaystyle  f(\mcal{Y}) = \max_{i \in \{1,2\}} R^{(0)}_{Si} + R^{(1)}_{33} + \max_{i \in \{1,2\}} R^{(2)}_{iD} $ & $2\log(2)$ \\ \hline
            (5) & \begin{tabular}{l}$\mcal{Y} = \{S,\{2,3\},\phi \}$\\$\mcal{Y}^r = \{S,\{1,2,3\},\{1\}\}$\end{tabular} &  \begin{tabular}{l} $\displaystyle  f(\mcal{Y}) = R^{(0)}_{S1} + \hspace{-0.05in}\max_{i \in \{2,3\}}\hspace{-0.05in}R^{(1)}_{ip} + M(1)_{\{2,3\}}^{\{1,2,3\}\backslash\{p\}}\ ,\quad\quad p \in \{1,2,3\} $ \\ $ \displaystyle  f(\mcal{Y}) = R^{(0)}_{S1} + \max_{\substack{\mb{v} \subset \{1,2,3\},\\ \vert \mb{v} \vert = 2} }M(1)_{\{2,3\}}^{\{\mb{v}\}}$\quad,\quad $  \displaystyle  f(\mcal{Y}) = R^{(0)}_{S1} + \hspace{-0.05in}\max_{i \in \{1,2,3\}}\hspace{-0.05in}R^{(1)}_{2i} + \hspace{-0.05in}\max_{i \in \{1,2,3\}} \hspace{-0.05in}R^{(1)}_{3i}$\quad  \end{tabular} & $2\log(3)$\\ \hline
            (6) & \begin{tabular}{l}$ \mcal{Y} = \{S,\{2,3\},\{1\} \}$\end{tabular} & \begin{tabular}{l}$\displaystyle f(\mcal{Y}) = R^{(0)}_{S1} +  \max_{i \in \{2,3\}} R^{(1)}_{2i} + \max_{i \in \{2,3\}} R^{(1)}_{3i} + R^{(2)}_{1D}$\ ,\\ $\displaystyle f(\mcal{Y}) =  R^{(0)}_{S1} + \max_{i \in \{2,3\}} R^{(1)}_{i2} + \max_{i \in \{2,3\}} R^{(1)}_{i3} + R^{(2)}_{1D}$\ ,\\ $\displaystyle f(\mcal{Y}) = R^{(0)}_{S1} + M(1)_{\{2,3\}}^{\{2,3\}} + R^{(2)}_{1D}$ \end{tabular} & $2\log(2)$ \\ \hline
            (7) & \begin{tabular}{l}$\mcal{Y} = \{S,\{1,2,3\},\phi \}$\end{tabular} & 
            \begin{tabular}{l}$\displaystyle f(\mcal{Y}) = \max_{i \in \{1,2,3\}} R^{(1)}_{pi} + \max_{\substack{\mb{v} \subset \{1,2,3\},\\ \vert \mb{v} \vert = 2}}M(1)_{\{1,2,3 \}\backslash\{p\}}^{\{\mb{v}\}}\ ,\quad\quad p \in \{1,2,3\}$ \\ $\displaystyle f(\mcal{Y}) = \frac{3}{2} \max_{ \substack{\mb{u} \subset \{1,2,3\},\\ \mb{v} \subset \{1,2,3\},\\ \vert \mb{u} \vert = \vert \mb{v} \vert = 2}}\hspace{-0.05in} M(1)_{\{\mb{u}\}}^{\{\mb{v}\}}$ \end{tabular} & $3\log(3)$ \\ \hline
            (8) & \begin{tabular}{c}$Otherwise$\end{tabular} & $\displaystyle f(\mcal{Y}) = M(0)_{\{S\}}^{\{\mb{y}_1^c\}} + M(1)_{\{\mb{y}_1\}}^{\{\mb{y}_2^c\}} + M(2)_{\{\mb{y}_2\}}^{\{D\}}$ & $\displaystyle zero$ \\[3pt] \hline

        \end{tabular}
        \egroup
    \end{table*}

For the remainder of this section, we define
\begin{align} \label{eq:k2_C_tilde}
    \widetilde{C} \triangleq \min_{\mcal{Y}} \min_{f \in \mcal{F}_\mcal{Y}} f(\mcal{Y})
\end{align}
and as a result
\begin{align}
    \label{eq:tilde_to_bar_relation}
    \widebar{C} \leq \widetilde{C} + 3\log(3)
\end{align}
To prove Theorem \ref{thm:k2_out_of_3}, we are going to argue by contradiction that for any network with $L=2$ and $N=3$, there always exists a subnetwork with two relays per layer such that $\widebar{C}_2$ for this subnetwork is greater than $\frac{1}{2}\widetilde{C}$.
Once this is established, the statement of the theorem follows directly as:
\begin{align*}
    C_2^* &\geq \frac{1}{2}\widetilde{C} \geq \frac{1}{2}\widebar{C} - 1.5\log(3)
\end{align*}

In the remainder of this proof, by a slight abuse of notation, we refer to a subnetwork comprising relays $(1,i_1)$, $(1,j_1)$ in the first layer and $(2,i_2), (2,j_2)$ in the second layer by the vector tuple $[\{i_1,j_1\},\{i_2,j_2\}]$.

We start by assuming that for an arbitrary network with $L=2$ and $N=3$, all its subnetworks with two relays per layer have $\widebar{C}_2$ less than $\frac{1}{2}\widetilde{C}$, i.e., every such subnetwork has at least one (if not more) cut(s) with cut-value less than $\frac{1}{2} \widetilde{C}$.
We term such cuts of the subnetworks as \textit{critical cuts}.
Let $\Lambda$ be the union of the links forming the critical cuts.
Also, since each link $(l,i,j) \in \Lambda$ is part of at least one critical cut, then $ \forall (l,i,j) \in \Lambda$, we have $(l,i,j) < \frac{1}{2}\widetilde{C}$.

We proceed by categorizing $\Lambda$ into classes, depending on how many links $(0,S,j)$ and $(2,i,D) \in \Lambda$.
Let $z_0 = \vert\{j : (0,S,j) \in \Lambda \}\vert$ and $z_2 = \vert \{ i : (2,i,D) \in \Lambda \} \vert$.
We therefore, need to address the following cases (the others follow from symmetry):\\
    (1) $z_0 = 0$ and $z_2 = 0$,\quad (2) $z_0 = 1$ and $z_2 = 0$\\
    (3) $z_0 = 1$ and $z_2 = 1$,\quad (4) $z_0 = 1$ and $z_2 = 2$\\
    (5) $z_0 = 2$ and $z_2 = 0$,\quad (6) $z_0 = 2$ and $z_2 = 2$ \\
    (7) $z_0 = 3$ or $z_2 = 3$

    Before we go through the proof for these cases, it is of benefit to discuss some simple implications which we use extensively throughout the proof.
\begin{enumerate}[1.]
    \item Assume that the link $(0,S,i) \not\in \Lambda$. Then for any subnetwork $[\{i,x\},\{y,z\}]$, the critical cut cannot contain a term of the form $M(0)_{\{S\}}^{\{i,\theta\}}$ where $\theta = x$ or $\theta = \phi$ (empty set). In other words, critical cuts of such subnetworks always consider relay $(1,i)$ to be on the source side (i.e., a transmitter). Similarly, if $(2,i,D) \not\in \Lambda$ then all critical cuts of subnetworks of the form $[\{x,y\},\{i,z\}]$ consider relay $(2,i)$ on the destination side (i.e., a receiver).\\

    \item Assume that a subnetwork has a critical cut capacity of the form: $X + M(l)_{\{\mb{u}\}}^{\{\mb{v}\}} < \frac{1}{2}\widetilde{C}$,
        where $X$ is an arbitrary term representing other contributions to the critical cut capacity ($X$ can be zero).
        Then, by the fact that the capacity of a MIMO channel is lower bounded by the capacity of any of its subchannels (subset of transmitters and/or receivers), we have:
        \begin{align}
            & & X + \max_{j \in \{\mb{v}\} } R^{(l)}_{tj} < \frac{1}{2}\widetilde{C}\quad \forall t \in \{\mb{u}\}    \nonumber \\
            &\text{and}&\  X + \max_{i \in \{\mb{u}\}} R^{(l)}_{it} < \frac{1}{2}\widetilde{C}\quad \forall t \in \{\mb{v}\}
            \label{eq:mimo_to_link}
        \end{align}
         While constructing our contradictions, we will always use these implications directly whenever we have a SIMO or MISO expression as part of the critical cut capacity.
        When we have a MIMO expression as the critical cut, we will selectively use \eqref{eq:mimo_to_link} and note it accordingly.
\end{enumerate}
We can now proceed to prove each of the cases listed.
\subsection{\fbox{$z_0 = 0$ and $z_2 = 0$}}
In this case $(0,S,i) \not\in \Lambda$ and $(2,i,D) \not\in \Lambda$, $\forall i \in \{1,2,3\}$. This means that for all subnetworks, only MIMO cuts are critical cuts.
Now consider the full network cut characterized by $\mcal{Y} = \{ S, \{1,2,3\}, \phi\}$.
We have the following contradiction:
            \begin{align*}
                \widetilde{C} \stackrel{(a)}\leq \min_{f \in \mcal{F}_{\mcal{Y}}} f(\mcal{Y}) \stackrel{(b)}{\leq}&\ \ \frac{3}{2} \max_{\substack{ \mb{u},\mb{v} \subseteq \{1,2,3\},\\ \vert \mb{u} \vert = \vert \mb{v} \vert = 2}} M(1)^{\mb{\{v\}}}_{\mb{\{u\}}}\\
                                                       <&\ \ \frac{3}{2} \times \frac{1}{2}\widetilde{C} = \frac{3}{4}\widetilde{C}
                                                   \end{align*}
where $(a)$ follows from \eqref{eq:k2_C_tilde} and $(b)$ follows from row 7 in Table \ref{tab:class_to_f_matching}.
\subsection{\fbox{$z_0 = 1$ and $ z_2 = 0$ }}
            Without loss of generality, we can assume that the link $(0,S,1) \in \Lambda$.
            Consider the subnetworks constructed by selecting relays $[ \{2,3\},\{s_1,s_2\}]$ where $s_1,s_2 \in \{1,2,3\}$ and $s_1 \neq s_2$.
            Since $z_0 = 1$, we know that $(0,S,i) \not\in \Lambda$ for $i \in \{2,3\}$.
            Additionally since $z_2 = 0$, this implies that for all aforementioned subnetworks, the critical cuts are only of the form:
            \[
                Cut(\mcal{Z}_1) = M(1)^{\{s_1,s_2\}}_{\{2,3\}} < \frac{1}{2} \widetilde{C}
            \]
            Since this is true for all three subnetworks characterized by $s_1,s_2$ as mentioned above, it implies that:
            \[
                \max_{\substack{s_1,s_2 \in \{1,2,3\},\\s_1 \neq s_2}} M(1)^{\{s_1,s_2\}}_{\{2,3\}} < \frac{1}{2}\widetilde{C}
            \]
            Now from the full network cut characterized by $\mcal{Y} = \{S, \{2,3\},\phi\}$, we have:
            \begin{align*}
                \widetilde{C} \stackrel{(a)}{\leq} \min_{f \in \mcal{F}_{\mcal{Y}}} f(\mcal{Y}) &\stackrel{(b)}{\leq} R_{S1}^{(0)} + \max_{\substack{s_1,s_2 \in \{1,2,3\},\\s_1 \neq s_2}} M(1)^{\{s_1,s_2\}}_{\{2,3\}} \\
                        &< \frac{1}{2}\widetilde{C} + \frac{1}{2}\widetilde{C} = \widetilde{C}
            \end{align*}
            which gives a contradiction. Note that $(a)$ follows from \eqref{eq:k2_C_tilde} and $(b)$ from row 5 in Table \ref{tab:class_to_f_matching}. A similar argument follows for the case when $z_0 = 0$ and $z_2 = 1$.
            \\
            \subsection{\fbox{$z_0 = 1$ and $z_2 = 1$} }
            Without loss of generality, we can assume that the links $(0,S,1),(2,1,D) \in \Lambda$.  Consider the subnetwork constructed by selecting the relays $[\{2,3\},\{2,3\}]$.
            For this subnetwork, the critical cut can only be the MIMO cut. i.e.
            \begin{equation}
                M(1)^{\{2,3\}}_{\{2,3\}} < \frac{1}{2} \widetilde{C}
                \label{eq:case_1_1_odd_sub}
            \end{equation}
            If in addition, we have:
            \[
                R^{(0)}_{S1} + R^{(2)}_{1D} < \frac{1}{2} \widetilde{C}
            \]
            then we get a contradiction, since for the cut characterized by $\mcal{Y} = \{ S,\{2,3\},\{1\}\}$, we have:
            \[
                \begin{aligned}
                    \widetilde{C} \leq \min_{f \in \mcal{F}_{\mcal{Y}}} f(\mcal{Y}) \stackrel{(a)}{\leq}& R^{(0)}_{S1} + M(1)_{\{2,3\}}^{\{2,3\}} + R^{(2)}_{1D}  \\
                    <& \frac{1}{2} \widetilde{C} + \frac{1}{2} \widetilde{C} = \widetilde{C}
                \end{aligned}
            \]
            where $(a)$ follows from row 6 in Table \ref{tab:class_to_f_matching}.
            As a result, from here onwards, we assume that:
            \begin{equation}
                R^{(0)}_{S1} + R^{(2)}_{1D} \geq \frac{1}{2} \widetilde{C}
                \label{eq:less_than_05}
            \end{equation}
            Consider the subnetworks formed by relays $[\{1,s_1\},\{2,3\}]$ where $s_1 \in \{2,3\}$.
            For these subnetworks parameterized by $s_1$, the candidate critical cuts can only be one of these types:
            \begin{equation}
                \begin{aligned}
                            \text{\textbf{Type I(a)}}  \\
                            Cut(\mcal{Z}_1)_{(s_1)} =& R^{(0)}_{S1} + \max_{i\in\{2,3\}} R^{(1)}_{s_1 i} < \frac{1}{2} \widetilde{C}\\
                            \text{\textbf{Type II(a)}}\\
                            Cut(\mcal{W}_1)_{(s_1)} =& M(1)_{\{1,s_1\}}^{\{2,3\}} < \frac{1}{2} \widetilde{C}\\
                \end{aligned}
                \label{eq:case_1_1_net1}
            \end{equation}
            Also consider the subnetworks $[\{2,3\},\{1,s_2\}]$ where $s_2 \in \{2,3\}$.
            The candidate cuts for these subnetworks can only be of the following categories:
            \begin{equation}
                \begin{aligned}
                            \text{\textbf{Type I(b)}}  \\
                            Cut(\mcal{Z}_2)_{(s_2)} =& \max_{i\in\{2,3\}} R^{(1)}_{i s_2} + R^{(2)}_{1D}  < \frac{1}{2} \widetilde{C}\\
                            \text{\textbf{Type II(b)}}\\
                            Cut(\mcal{W}_2)_{(s_2)} =& M(1)_{\{2,3\}}^{\{1,s_2\}}  < \frac{1}{2} \widetilde{C}\\
                \end{aligned}
                \label{eq:case_1_1_net2}
            \end{equation}
            Of the four subnetworks described above (parameterized by $s_1$ and $s_2$), if three or more subnetworks have critical cuts of Type I, then at least two are of Type I(a), else two are of Type I(b).
            From \eqref{eq:case_1_1_net1} and \eqref{eq:case_1_1_net2}, this implies that either:
            \[
                \begin{aligned}
                R^{(0)}_{S1} + \max_{i,j \in \{2,3\}} R^{(1)}_{ij} &< \frac{1}{2} \widetilde{C}, \\
                \max_{i \in \{2,3\}} R^{(1)}_{i s_2} + R^{(2)}_{1D} &< \frac{1}{2} \widetilde{C}, \text{ for some } s_2 \in \{2,3\} \\
                \text{\bf or}&    \\
                \max_{i,j \in \{2,3\}} R^{(1)}_{ij} + R^{(2)}_{1D} &< \frac{1}{2} \widetilde{C}, \\
                R^{(0)}_{S1} + \max_{i \in \{2,3\}} R^{(1)}_{s_1 i} &< \frac{1}{2} \widetilde{C}, \text{ for some } s_1 \in \{2,3\} \\
                \end{aligned}
            \]
            Let $\hat{s}_k \in \{2,3\}, \hat{s}_k \neq s_k$ for $k \in \{1,2\}$.
            Now considering the cut $\mcal{Y} = \{S,\{2,3\},\{1\} \}$, we arrive at a contradiction as follows:
            \begin{align*}
                    \widetilde{C} &\leq \min_{f \in \mcal{F}_{\mcal{Y}}} f(\mcal{Y})\\
                                  &\stackrel{(a)}{\leq} R^{(0)}_{S1} + \max_{i \in \{2,3\}} R^{(1)}_{i \hat{s}_2} + \max_{i \in \{2,3\}} R^{(1)}_{i s_2} + R^{(2)}_{D1} \\
                                  &\leq \underbracket{R^{(0)}_{S1} + \max_{i,j \in \{2,3\}} R^{(1)}_{ij}} + \underbracket{\max_{i \in \{2,3\}} R^{(1)}_{i s_2} + R^{(2)}_{D1}} \\
                                  &< \frac{1}{2} \widetilde{C} + \frac{1}{2} \widetilde{C} = \widetilde{C} \\
                & \text{\bf or}   \\
                    \widetilde{C} &\leq \min_{f \in \mcal{F}_{\mcal{Y}}} f(\mcal{Y})\\
                                  &\stackrel{(b)}{\leq} R^{(0)}_{S1} + \max_{i \in \{2,3\}} R^{(1)}_{s_1 i} + \max_{i \in \{2,3\}} R^{(1)}_{\hat{s}_1 i} + R^{(2)}_{D1} \\
                                  &\leq \underbracket{R^{(0)}_{S1} +\max_{i \in \{2,3\}} R^{(1)}_{s_1 i}} + \underbracket{\max_{i,j \in \{2,3\}} R^{(1)}_{ij} + R^{(2)}_{D1}} \\
                                  &< \frac{1}{2} \widetilde{C} + \frac{1}{2} \widetilde{C} = \widetilde{C}
            \end{align*}
            where $(a)$ and $(b)$ follow from row 6 in Table \ref{tab:class_to_f_matching}.

            Similarly, if three or more of the subnetworks (parameterized by $s_1$ and $s_2$) have Type II critical cuts, then at least two of them are of Type II(a) or else we have two of Type II(b).
            Either way, this implies that:
            \begin{equation}
                \begin{aligned}
                    \max_{i \in \{2,3\}} M(1)_{\{1,i\}}^{\{2,3\}} < \frac{1}{2}\widetilde{C}\\
                    \text{or}\ \max_{i \in \{2,3\}} M(1)_{\{2,3\}}^{\{1,i\}} < \frac{1}{2} \widetilde{C}
                \end{aligned}
                \label{eq:case_1_1_MIMOs}
            \end{equation}
            Now consider the full network cuts $\mcal{Y}_1 = \{S,\{2,3\},\phi \}$ and $\mcal{Y}_2 = \{S, \{1,2,3\},\{1\}\}$.
            One of these cuts gives us a contradiction as follows:
            \[
                \begin{aligned}
                    \widetilde{C} &\leq \min_{f \in \mcal{F}_{\mcal{Y}_1}} f(\mcal{Y}_1)\\
                    &\stackrel{(a)}{\leq} R^{(0)}_{S1} + \max_{\substack{i,j \in \{1,2,3\},\\ i\neq j}} M(1)^{\{i,j\}}_{\{2,3\}} \\
                                  &\stackrel{(b)}{<} \frac{1}{2} \widetilde{C} + \frac{1}{2} \widetilde{C} = \widetilde{C} \\
                & \text{\bf or}   \\
                    \widetilde{C} &\leq \min_{f \in \mcal{F}_{\mcal{Y}_2}} f(\mcal{Y}_2)\\
                    &\stackrel{(a)}{\leq} \max_{\substack{i,j \in \{1,2,3\},\\ i\neq j}} M(1)_{\{i,j\}}^{\{2,3\}} + R^{(2)}_{1D} \\
                                  &\stackrel{(b)}{<} \frac{1}{2} \widetilde{C} + \frac{1}{2} \widetilde{C} = \widetilde{C} \\
                \end{aligned}
            \]
            where $(a)$ in both occurrences\footnote{In the case of $\mcal{Y}_2$, the expressions in the Table need to be accordingly modified to address a ``$3 \times 2$ MIMO + link'' instead of ``$2 \times 3$ MIMO + link''.} follow from row 5 in Table \ref{tab:class_to_f_matching} and $(b)$ in both cases is implied by \eqref{eq:less_than_05} and \eqref{eq:case_1_1_MIMOs}.

            Note that the argument just mentioned is valid even if only two subnetworks have Type II critical cuts as long as they are both Type II(a) or Type II(b).
            Therefore, the only scenario to reconcile with, is when for the four subnetworks parameterized by $s_1$ and $s_2$, we have one critical cut from each Type: I(a), I(b), II(a) and II(b).

            Again let $\hat{s}_k \in \{2,3\}, \hat{s}_k \neq s_k$ for $k \in \{1,2\}$.
            The remaining scenario can be represented by any combination of $s_1,s_2,\hat{s}_1,\hat{s}_2$, and we have:
            \begin{subequations}
                \begin{align}
                    R^{(0)}_{S1} + \max_{i \in \{2,3\}} R^{(1)}_{s_1 i} <& \frac{1}{2}\widetilde{C} \label{eq:case_1_1_final_1}\\
                    M(1)^{\{2,3\}}_{\{1,\hat{s}_1\}} <& \frac{1}{2} \widetilde{C} \label{eq:case_1_1_final_2}\\
                    \max_{i \in \{2,3\}} R^{(1)}_{i s_2} + R^{(2)}_{1D} <& \frac{1}{2} \widetilde{C}  \label{eq:case_1_1_final_3}\\
                    M(1)_{\{2,3\}}^{\{1,\hat{s}_2\}} <& \frac{1}{2} \widetilde{C} \label{eq:case_1_1_final_4}
                \end{align}
                \label{eq:case_1_1_final_situation}
            \end{subequations}
            From \eqref{eq:case_1_1_final_situation}, we can conclude the following:
            \begin{equation}
                \begin{aligned}
                \max_{i \in \{1,2,3\}} R^{(1)}_{ki} &< \frac{1}{2}\widetilde{C},\quad \forall k \in \{2,3\} \\
                \max_{i \in \{1,2,3\}} R^{(1)}_{ik} &< \frac{1}{2}\widetilde{C},\quad \forall k \in \{2,3\}
                \end{aligned}
                \label{eq:case_1_1_final_situation_conc}
            \end{equation}
            We cannot directly argue a contradiction using the relations in  \eqref{eq:case_1_1_final_situation} and \eqref{eq:case_1_1_final_situation_conc}.
            Therefore, we consider the remaining subnetworks constructed from relays $[\{1,u_1\},\{1,u_2\}]$, where $u_1,u_2 \in \{2,3\}$.
            For this subnetwork parameterized by $u_1,u_2$, there are four types of candidate critical cuts.
            One possible critical cut is
            \[
                Cut(\mcal{T}_3)_{(u_1,u_2)} = R^{(0)}_{S1} + R^{(1)}_{u_1 u_2} + R^{(2)}_{1D} \\
            \]
            However since starting \eqref{eq:less_than_05}, we assume that
            \[
                R^{(0)}_{S1} +  R^{(2)}_{1D} \geq \frac{1}{2} \widetilde{C}
            \]
            then $Cut(\mcal{T}_3)$ cannot be a critical cut.
            The three remaining candidate critical cuts are:
            \begin{equation}
                \begin{aligned}
                    \text{\textbf{Type I(c)}}  \\
                            Cut(\mcal{V}_3)_{(u_1,u_2)} =& R^{(0)}_{S1} + \max_{i\in\{1,u_2\}} R^{(1)}_{u_1 i} <\frac{1}{2} \widetilde{C} \\
                            \text{\textbf{Type II(c)}}\\
                            Cut(\mcal{Z}_3)_{(u_1,u_2)} =& \max_{i\in\{1,u_1\}} R^{(1)}_{i u_2} + R^{(2)}_{1D}  <\frac{1}{2} \widetilde{C} \\
                            \text{\textbf{Type III(c)}}\\
                            Cut(\mcal{W}_3)_{(u_2,u_2)} =& M(1)^{\{1,u_2\}}_{\{1,u_1\}}  <\frac{1}{2} \widetilde{C} \\
                \end{aligned}
                \label{eq:case_1_1_rem_networks}
            \end{equation}
            Now we are going to prove that if any of the subnetworks characterized by $u_1, u_2$ have a critical cut of Type I(c), then we get a contradiction.
            For that, we have one of the following scenarios:
            \begin{enumerate}[(1) ]
                \item \underline{$u_1 = s_1$ and $u_2 \in \{2,3\}$}:\\
                        In this case, we have from \eqref{eq:case_1_1_final_1} and \eqref{eq:case_1_1_rem_networks}:
                       \begin{align*}
                            & R^{(0)}_{S1} + \max_{i \in \{1,2,3\}} R^{(1)}_{s_1 i}\\
                            &= R^{(0)}_{S1} + \max\left\{  \max_{i \in \{1,u_2\}} R^{(1)}_{s_1 i},\ \max_{i \in \{2,3\}} R^{(1)}_{s_1 i} \right\} \\
                            &< \frac{1}{2} \widetilde{C} \numberthis
                            \label{eq:case_1_1_u1_eq_x1}
                        \end{align*}
                        Now considering the full network cut $\mcal{Y} = \{S,\{2,3\},\phi\}$, we have:
                        \[
                            \begin{aligned}
                                \widetilde{C} \leq \min_{f \in \mcal{F}_{\mcal{Y}}} f(\mcal{Y}) &\stackrel{(a)}{\leq} R^{(0)}_{S1} +\hspace{-0.05in} \max_{ i \in \{1,2,3\}} R^{(1)}_{2i} + \hspace{-0.05in}\max_{ i \in \{1,2,3\}} R^{(1)}_{3i} \\
                                &=  \underbracket{R^{(0)}_{S1} +\hspace{-0.05in} \max_{ i \in \{1,2,3\}} R^{(1)}_{s_1 i}} + \hspace{-0.05in}\max_{ i \in \{1,2,3\}} R^{(1)}_{\hat{s}_1 i} \\
                                                                 &\stackrel{(b)}{<} \frac{1}{2}\widetilde{C} + \frac{1}{2}\widetilde{C} = \widetilde{C}
                            \end{aligned}
                        \]
                        which is a contradiction. The relation $(a)$ uses row 5 in Table \ref{tab:class_to_f_matching} and $(b)$ follows from \eqref{eq:case_1_1_final_situation_conc} and \eqref{eq:case_1_1_u1_eq_x1}.
                        \\
                    \item \underline{$u_1 = \hat{s}_1$ and $u_2 = s_2$}:\\
                        In this scenario, we have:
                            \begin{align*}
                                &R^{(0)}_{S1} + \max_{i \in \{2,3\}} R^{(1)}_{i s_2} \\
                                &\leq R^{(0)}_{S1} + \max \left\{ \max_{i \in \{1,s_2\}} R^{(1)}_{\hat{s}_1 i}\ , R^{(1)}_{s_1 s_2} \right\} \\
                                &\leq R^{(0)}_{S1} + \max \left\{ \max_{i \in \{1,s_2\}} R^{(1)}_{\hat{s}_1 i}\ , \max_{i \in \{2,3\}} R^{(1)}_{s_1 i} \right\} \\
                                &= R^{(0)}_{S1} + \max \left\{ \max_{i \in \{1,u_2\}} R^{(1)}_{u_1 i}\ , \max_{i \in \{2,3\}} R^{(1)}_{s_1 i} \right\} \\
                                &\stackrel{(a)}{<} \frac{1}{2}\widetilde{C} \numberthis
                            \label{eq:case_1_1_u1_neq_x1_u2_eq_x2}
                        \end{align*}
                    where $(a)$ follows from \eqref{eq:case_1_1_final_1} and \eqref{eq:case_1_1_rem_networks}.
                        Considering the full network cut $\mcal{Y} = \{S, \{2,3\},\phi\}$, we are faced with a contradiction:
                        \[
                            \begin{aligned}
                                \widetilde{C} \leq \min_{f \in \mcal{F}_{\mcal{Y}}} f(\mcal{Y}) &\stackrel{(a)}{\leq}  \underbracket{R^{(0)}_{S1} +\hspace{-0.05in} \max_{ i \in \{2,3\}} R^{(1)}_{i s_2}} + M(1)_{\{2,3\}}^{\{1,\hat{s}_2\}} \\
                                 &\stackrel{(b)}{<} \frac{1}{2}\widetilde{C} + \frac{1}{2}\widetilde{C} = \widetilde{C}
                            \end{aligned}
                        \]
                        Relation $(a)$ uses row 5 from Table \ref{tab:class_to_f_matching} while $(b)$ follows from \eqref{eq:case_1_1_final_4} and \eqref{eq:case_1_1_u1_neq_x1_u2_eq_x2}.
                        \\
                    \item \underline{$u_1 = \hat{s}_1$ and $u_2 = \hat{s}_2$}:\\
                        In this case, we have:
                        \begin{align*}
                                &R^{(0)}_{S1} + \max_{i \in \{2,3\}} R^{(1)}_{i \hat{s}_2} \\
                                &= R^{(0)}_{S1} \hspace{-0.02in} + \hspace{-0.02in}\max\hspace{-0.02in}\left\{ R^{(1)}_{\hat{s}_1 \hat{s}_2}\ , R^{(1)}_{s_1 \hat{s}_2} \right\} \\
                                &\leq R^{(0)}_{S1} \hspace{-0.02in} + \hspace{-0.02in}\max\hspace{-0.02in}\left\{ \max_{i \in \{1,\hat{s}_2\} } R^{(1)}_{\hat{s}_1,i}\ , \max_{i \in \{2,3\}} R^{(1)}_{s_1,i} \right\} \\
                                &\stackrel{(a)}{<} \frac{1}{2}\widetilde{C} \numberthis
                            \label{eq:case_1_1_u1_neq_x1_u2_neq_x2}
                            \end{align*}
                            where $(a)$ follow from \eqref{eq:case_1_1_final_1} and \eqref{eq:case_1_1_rem_networks}.

                        For the full network cut $\mcal{Y} =\{S,\{2,3\},\{1\}\}$, we have:
                        \[
                            \begin{aligned}
                                \widetilde{C} &\leq \min_{f \in \mcal{F}_{\mcal{Y}}} f(\mcal{Y})\\
                                              &\stackrel{(a)}{\leq}  \underbracket{R^{(0)}_{S1} + \max_{ i \in \{2,3\}} R^{(1)}_{i \hat{s}_2}} + \underbracket{\max_{ i \in \{2,3\}} R^{(1)}_{i s_2} + R^{(2)}_{1D}} \\
                                              &\stackrel{(b)}{<} \frac{1}{2}\widetilde{C} + \frac{1}{2}\widetilde{C} = \widetilde{C}
                            \end{aligned}
                        \]
                        where $(a)$ follows from row 6 in Table \ref{tab:class_to_f_matching} and $(b)$ follows from \eqref{eq:case_1_1_final_3} and \eqref{eq:case_1_1_u1_neq_x1_u2_neq_x2}.
            \end{enumerate}
            If one subnetwork has a critical cut of Type II(c), similar arguments to the ones used to prove the contradictions above, can be made.
            This is due to the symmetry of the cuts in \eqref{eq:case_1_1_final_situation} and the symmetry of Type I and Type II critical cuts in \eqref{eq:case_1_1_rem_networks}.

            As a result, the only remaining scenario to consider is if in addition to \eqref{eq:case_1_1_final_situation}, all the subnetworks characterized by $u_1, u_2$ in \eqref{eq:case_1_1_rem_networks} have a critical cut of Type III(c).

            From \eqref{eq:case_1_1_rem_networks}, consider only the subnetworks where $u_1 = \hat{s}_1$, i.e. the two subnetworks with parameters $(u_1, u_2) = (\hat{s}_1,2)$ and $(u_1,u_2) = (\hat{s}_1,3)$.
            Combining the cuts from these subnetworks with \eqref{eq:case_1_1_final_4}, we can conclude that:
                \begin{align*}
                    &\max_{\substack{i,j \in \{1,2,3\},\\ i \neq j}} M(1)^{\{1,\hat{s}_2\}}_{\{i,j\}} \\
                    &= \max \left\{ M(1)^{\{1,u_2=2\}}_{\{1,\hat{s}_1\}},\ M(1)^{\{1,u_2=3\}}_{\{1,\hat{s}_1\}},\ M(1)^{\{2,3\}}_{\{1,\hat{s}_1\}} \right\} \\
                    &< \frac{1}{2} \widetilde{C} \numberthis
                \label{eq:case_1_1_rem_all_MIMO}
            \end{align*}
            From this conclusion, we can build a contradiction by considering the full network cut $\mcal{Y} = \{S,\{1,2,3\},\phi \}$:
            \[
                \begin{aligned}
                                \widetilde{C} &\leq \min_{f \in \mcal{F}_{\mcal{Y}}} f(\mcal{Y})\\
                                &\stackrel{(a)}{\leq}  \max_{i \in \{1,2,3\}} R^{(1)}_{s_1 i} + \max_{\substack{i,j \in \{1,2,3\},\\ i \neq j}} M(1)_{\{1,\hat{s}_1\}}^{\{i,j\}} \\
                                &\stackrel{(b)}{<}  \frac{1}{2}\widetilde{C} + \frac{1}{2}\widetilde{C}  = \widetilde{C}
                \end{aligned}
            \]
            Relation $(a)$ follows from row 7 in Table \ref{tab:class_to_f_matching} while $(b)$ follows from \eqref{eq:case_1_1_final_situation_conc} and \eqref{eq:case_1_1_rem_all_MIMO}.

            This concludes the proof for $z_0 = 1$ and $z_2 = 1$.
            \\

    \subsection{\fbox{$z_0 = 1$ and $z_2 = 2$}}
            Without loss of generality, assume that the links $(0,S,1), (2,i,D) \in \Lambda$ for $i=1,2$.
            Consider the subnetwork with relays $[\{2,3\},\{2,3\}]$.
            For this subnetwork, the critical cuts are one of the following:
            \[
               \begin{aligned}
                   Cut(\mcal{Z}_1) =& M(1)_{\{2,3\}}^{\{2,3\}} < \frac{1}{2}\widetilde{C} \\
                Cut(\mcal{Z}_2) =& \max_{i \in \{2,3\}} R^{(1)}_{i3} + R^{(2)}_{2D} < \frac{1}{2}\widetilde{C}
                \end{aligned}
            \]
            Other cuts in the subnetwork are not critical or else we would have $z_0 \neq 1$ or $z_2 \neq 2$.
            Both conditions above, imply that:
            \begin{equation}
                \max_{i \in \{2,3\}} R^{(1)}_{i3} < \frac{1}{2} \widetilde{C}
                \label{eq:case_1_2_odd_sub}
            \end{equation}
            Now consider the subnetworks formed by selecting relays $[(1,s_1),(s_2,3)]$, where $s_1 \in \{2,3\}$ and $s_2 \in \{1,2\}$.
            We can classify the possible critical cuts in these subnetworks into two types.
            \begin{equation}
                        \begin{aligned}
                            \text{\textbf{Type I}}  \\
                            Cut(\mcal{Z}_1)_{(s_1,s_2)} =& M(1)_{\{1,s_1\}}^{\{s_2,3\}} < \frac{1}{2}\widetilde{C}\\
                            Cut(\mcal{Z}_2)_{(s_1,s_2)} =& \max_{i\in\{1,s_1\}} R^{(1)}_{i3} + R^{(2)}_{s_2 D}  < \frac{1}{2}\widetilde{C}\\
                            \text{\textbf{Type II}}\\
                            Cut(\mcal{W}_1)_{(s_1,s_2)} =& R^{(0)}_{S1} + R^{(1)}_{s_13} + R^{(2)}_{s_2D}  < \frac{1}{2}\widetilde{C}\\
                            Cut(\mcal{W}_2)_{(s_1,s_2)} =& R^{(0)}_{S1} + \max_{i\in\{s_2,3\}} R^{(1)}_{s_1i}  < \frac{1}{2}\widetilde{C}\\
                        \end{aligned}
                \label{eq:case_1_2_w}
            \end{equation}
            Note that for the subnetwork with relays $[\{1,s_1\},\{s_2,3\}]$, any critical cuts aside from the ones listed above would imply that either $z_0 \neq 1$ or $z_2 \neq 2$.
            For any of these four subnetworks parameterized by $s_1,s_2$, a critical cut of Type I implies from \eqref{eq:case_1_2_w} that:
            \[
                R^{(1)}_{13} < \frac{1}{2} \widetilde{C}
            \]
            By considering also the implication in \eqref{eq:case_1_2_odd_sub}, we have:
            \begin{equation}
                \max_{i \in \{1,2,3\}} R^{(1)}_{i3} < \frac{1}{2} \widetilde{C}
                \label{eq:case_1_2_t1}
            \end{equation}
            A critical cut of Type II implies that:
            \begin{equation}
                R^{(0)}_{S1} + R^{(1)}_{s_2 3} < \frac{1}{2}\widetilde{C}
                \label{eq:case_1_2_t2}
            \end{equation}

            If any of the four subnetworks parameterized by $s_1,s_2$ has a Type I critical cut, then we can get a contradiction as follows: Consider the full network cut $\mcal{Y} = \{S,\{1,2,3\},\{1,2\}\}$, we have:
            \[
                \begin{aligned}
                    \widetilde{C} \leq \min_{f \in \mcal{F}_{\mcal{Y}}} f(\mcal{Y}) &\stackrel{(a)}{\leq} \max_{i \in \{1,2,3\}} R^{(1)}_{i3} + \max_{i \in \{1,2\}} R^{(2)}_{iD} \\
                                                     &\stackrel{(b)}{<} \frac{1}{2}\widetilde{C} + \frac{1}{2}\widetilde{C} = \widetilde{C}
                \end{aligned}
            \]
            where the structure in $(a)$ follows from row 2 in Table \ref{tab:class_to_f_matching}. Relation $(b)$ follows from \eqref{eq:case_1_2_t1}, and the fact (remarked earlier) that $z_0 = 2$ coupled with our assumption $(2,i,D) \in \Lambda$ for $i=1,2$ imply that $R^{(2)}_{iD} < \frac{1}{2}\widetilde{C}$ for $i \in \{1,2\}$.

            If on the other hand, none of the four subnetworks have a critical cut of Type I, then they all have critical cuts of Type II.
            Consider \eqref{eq:case_1_2_t2} for all four subnetworks. Since $s_2 \in \{2,3\}$, we have:
            \[
                R^{(0)}_{S1} + \max_{i \in \{2,3\} } R^{(1)}_{i 3} < \frac{1}{2}\widetilde{C}
            \]
            We thus have a contradiction by considering the full network cut $\mcal{Y} = \{S,\{2,3\},\{1,2\}\}$ as follows:
            \[
                \begin{aligned}
                \widetilde{C} \leq \min_{f \in \mcal{F}_{\mcal{Y}}} f(\mcal{Y}) &\stackrel{(a)}{\leq} R^{(0)}_{S1} + \max_{i \in \{2,3\}} R^{(1)}_{i3} + \max_{i \in \{1,2\}} R^{(2)}_{iD} \\
                    &< \frac{1}{2}\widetilde{C} + \frac{1}{2}\widetilde{C} = \widetilde{C}
                \end{aligned}
            \]
            where $(a)$ follows the structure of row 3 in Table \ref{tab:class_to_f_matching}.
            This concludes the proof for the case $z_0 = 1, z_2 =2$. A similar argument can be made for $z_0 =2, z_2 = 1$.
\\

\subsection{\fbox{$ z_0 = 2$ and $z_2 = 0$ }}
            Without loss of generality, assume that the links $(0,S,i) \in \Lambda$ for $i \in \{1,2\}$, i.e.
            \begin{equation}
                \max_{i \in \{1,2\}} R^{(0)}_{Si} < \frac{1}{2}\widetilde{C}
                \label{eq:case_2_0_first_eq}
            \end{equation}
            Consider the subnetworks formed by selecting relays $[\{1,3\},\{s_1,s_2\}]$ where $s_1,s_2 \in \{1,2,3\}$ and $\ s_1 \neq s_2$.
            Since the links $(2,D,i) \not\in \Lambda$ for $i \in \{1,2,3\}$, then for these subnetworks (parameterized by $s_1, s_2$), we only have two candidates for critical cuts:
            \[
                \begin{aligned}
                    Cut(\mcal{Z})_{(s_1,s_2)} &= R^{(0)}_{S1} + \max_{i \in \{s_1,s_2\}} R^{(1)}_{3i} < \frac{1}{2} \widetilde{C} \\
                    Cut(\mcal{W})_{(s_1,s_2)} &= M_{\{1,3\}}^{\{s_1,s_2\}} < \frac{1}{2} \widetilde{C}
                \end{aligned}
            \]
            If either cut is critical, we arrive at the same conclusion:
            \begin{equation}
                \max_{i \in \{s_1,s_2\} } R^{(1)}_{3i} < \frac{1}{2} \widetilde{C}
                \label{eq:case_2_0_second_eq}
            \end{equation}
            By combining \eqref{eq:case_2_0_second_eq} above, for the subnetworks $(s_1, s_2) = (1,2)$ and $(s_1,s_2) = (1,3)$, we get:
            \[
                \max_{i \in \{1,2,3\} } R^{(1)}_{3i} < \frac{1}{2} \widetilde{C}
            \]
            Now considering the full network cut $\mcal{Y} = \{S,\{3\},\phi\}$, the following contradiction arises:
            \[
                \begin{aligned}
                    \widetilde{C} \leq \min_{f \in \mcal{F}_{\mcal{Y}}} f(\mcal{Y}) &\stackrel{(a)}{\leq} \max_{i \in \{1,2\}} R^{(0)}_{Si} + \max_{i \in \{1,2,3\} } R^{(1)}_{3i}\\
                    &\stackrel{(b)}{<} \frac{1}{2}\widetilde{C} + \frac{1}{2}\widetilde{C} = \widetilde{C}
                \end{aligned}
            \]
            where $(a)$ uses row 2 from Table \ref{tab:class_to_f_matching} while $(b)$ follows from \eqref{eq:case_2_0_first_eq} and \eqref{eq:case_2_0_second_eq}.

            This concludes the proof for the case $z_0 = 2$ and $z_2 = 0$.
            A similar argument follows for the case $z_0 = 0$ and $z_2 = 2$.

            \subsection{\fbox{$ z_0 = 2$ and $z_2 = 2$ }}
            Without loss of generality, assume that the links $(0,S,i),(2,i,D) \in \Lambda$ for $i \in \{1,2\}$,
            i.e., we have:
            \begin{equation}
                \max_{i \in \{1,2\}} R^{(0)}_{Si} < \frac{1}{2}\widetilde{C}\quad \text{and}\quad \max_{i \in \{1,2\}}R^{(2)}_{iD} < \frac{1}{2}\widetilde{C}
                \label{eq:case_2_2_assump}
            \end{equation}
            Consider the subnetworks constructed by selecting relays $[ \{s_1,3\},\{s_2,3\} ]$ for $s_1,s_2 \in \{1,2\}$.
            The candidate critical cuts for such subnetwork can be classified into three types as:
            \begin{equation}
                        \begin{aligned}
                            \text{\textbf{Type I}}\\
                            Cut(\mcal{V}_1)_{(s_1,s_2)} =& M(1)_{\{s_1,3\}}^{\{s_2,3\}} < \frac{1}{2}\widetilde{C} \\
                            \text{\textbf{Type II}}  \\
                            Cut(\mcal{Z}_1)_{(s_1,s_2)} =& R^{(0)}_{S s_1} + R^{(1)}_{33} + R^{(2)}_{s_2 D} < \frac{1}{2} \widetilde{C}\\
                            \text{\textbf{Type III}}\\
                            Cut(\mcal{W}_1)_{(s_1,s_2)} =& R^{(0)}_{S s_1} + \max_{i \in \{s_2,3\}} R^{(1)}_{3i} < \frac{1}{2} \widetilde{C} \\
                            Cut(\mcal{W}_2)_{(s_1,s_2)} =& \max_{i \in \{s_1,3\}} R^{(1)}_{i3} + R^{(2)}_{s_2 D} < \frac{1}{2} \widetilde{C} \\
                        \end{aligned}
                        \label{eq:case_2_2_w}
            \end{equation}
            If a network parameterized by $s_1, s_2$ has a Type I critical cut, then we have the following implications:
            \begin{equation}
                \begin{aligned}
                    \max_{i \in \{s_2,3\}} R^{(1)}_{3i} < \frac{1}{2} \widetilde{C}\ &\text{\textbf{and}}\ \max_{i \in \{s_1,3\}} R^{(1)}_{i3} < \frac{1}{2} \widetilde{C} \\
                \end{aligned}
            \label{eq:case_2_2_t1}
            \end{equation}
            If a subnetwork has a critical cut of Type II, this implies that:
            \begin{equation}
                \begin{aligned}
                    R^{(0)}_{S s_1} + R^{(1)}_{33} < \frac{1}{2}\widetilde{C}\quad ,&\quad  R^{(1)}_{33} + R^{(2)}_{s_2 D} < \frac{1}{2}\widetilde{C}\ ,\\
                    R^{(0)}_{S s_1} + R^{(2)}_{s_2 D} &< \frac{1}{2}\widetilde{C}
                \end{aligned}
                \label{eq:case_2_2_t2}
            \end{equation}

            Finally, a critical cut of Type III implies that either:
            \begin{equation}
                \max_{i \in \{s_2,3\}} R^{(1)}_{3i} < \frac{1}{2} \widetilde{C}\quad \text{\textbf{and/or}} \max_{i \in \{s_1,3\}} R^{(1)}_{i3} < \frac{1}{2} \widetilde{C}
                \label{eq:case_2_2_t3}
            \end{equation}
            Let $\hat{s}_k \in \{2,3\}$ such that $s_k \neq \hat{s}_k$ for $k \in \{1,2\}$.
            If a subnetwork with parameters $s_1, s_2$ has a critical cut of Type I, consider also the subnetwork with parameters $\hat{s}_1$ and $\hat{s}_2$.
            Since the second subnetwork can have critical cuts from any of the three types, then combining the conclusion from both subnetworks, we have from \eqref{eq:case_2_2_t1}, \eqref{eq:case_2_2_t2} and \eqref{eq:case_2_2_t3}, the following cases:
            \[
                \begin{aligned}
                    \text{(1):}\ & \max_{i \in \{1,2,3\}} R^{(1)}_{3i} < \frac{1}{2}\widetilde{C} \\
                    \text{(2):}\ & M(1)_{\{s_1,3\}}^{\{s_2,3\}} < \frac{1}{2}\widetilde{C} \\
                    &\text{\textbf{and}}\ R^{(0)}_{S \hat{s}_1} + R^{(2)}_{\hat{s}_2 D} < \frac{1}{2}\widetilde{C} \\
                    \text{(3):}\ & \max_{i \in \{1,2,3\}} R^{(1)}_{3i} < \frac{1}{2}\widetilde{C}\\
                    &\text{\textbf{or}} \max_{i \in \{1,2,3\}} R^{(1)}_{i3} < \frac{1}{2}\widetilde{C} \\
                \end{aligned}
            \]
            In the three cases, we get a contradiction from one of the full network cuts.
            In the first case, consider the cut characterized by $\mcal{Y}_1 = \{S,\{3\},\phi\}$. From this cut we get:
            \[
                \begin{aligned}
                    \widetilde{C} \leq \min_{f \in \mcal{F}_{\mcal{Y}}} f(\mcal{Y}) \stackrel{(a)}{\leq}& \max_{i \in \{1,2\} }R^{(0)}_{Si} + \max_{i \in \{1,2,3\}} R^{(1)}_{3i} \\
                                                     <&\ \frac{1}{2}\widetilde{C} + \frac{1}{2} \widetilde{C} = \widetilde{C}
                \end{aligned}
            \]
            which is a contradiction. Relation $(a)$ follows from row 2 in Table \ref{tab:class_to_f_matching}. In the second case, consider the cut $\mcal{Y}_2 = \{S,\{s_1,3\},\{\hat{s}_2\}\}$ depending on the choices of $s_1$ and $\hat{s}_2$.
            From this we get:
            \[
                \begin{aligned}
                    \widetilde{C} \leq \min_{f \in \mcal{F}_{\mcal{Y}}} f(\mcal{Y}) \stackrel{(b)}{\leq}& R^{(0)}_{S \hat{s}_1} + M(1)_{\{s_1,3\}}^{\{s_2,3\}} + R^{(2)}_{\hat{s}_2 D} \\
                                                      <& \frac{1}{2} \widetilde{C} + \frac{1}{2} \widetilde{C}
                \end{aligned}
            \]
            which is a contradiction. $(b)$ uses row 6 in Table \ref{tab:class_to_f_matching}. Finally for the third case, we get the contradiction either from the full network cut $\mcal{Y}_1$ again or from $\mcal{Y}_3 = \{S, \{1,2,3\},\{1,2\}\}$.
            The cut $\mcal{Y}_3$ gives us:
            \[
                \begin{aligned}
                    \widetilde{C}  \leq \min_{f \in \mcal{F}_{\mcal{Y}}} f(\mcal{Y}) \stackrel{(c)}{\leq}& \max_{i \in \{1,2,3\}} R^{(1)}_{i3} + \max_{i \in \{1,2\} }R^{(2)}_{iD} \\
                                                     <&\ \frac{1}{2}\widetilde{C} + \frac{1}{2} \widetilde{C} = \widetilde{C}
                \end{aligned}
            \]
            which is a contradiction, where $(c)$ is implied by row 2 in Table \ref{tab:class_to_f_matching}.

            We now turn our attention to subnetworks with Type II and Type III critical cuts.
            Since two subnetworks would differ at least in one of the two parameters $s_1, s_2$, it follows that having two or more subnetworks with critical cuts of Type II implies from \eqref{eq:case_2_2_t2} that either:
            \begin{equation}
                \begin{aligned}
                    &\max_{i \in \{1,2\} }R^{(0)}_{Si} + R^{(1)}_{33} &< \frac{1}{2}\widetilde{C}\\
                    &\text{\textbf{or}}& \\
                &R^{(1)}_{33} + \max_{i \in \{1,2\}}R^{(2)}_{iD} &< \frac{1}{2}\widetilde{C}
                \end{aligned}
                \label{eq:case_2_2_t1_conc}
            \end{equation}
            By considering the cut $\mcal{Y} = \{ S,\{3\},\{1,2\}\}$, we have the following contradiction:
            \[
                \begin{aligned}
                    \widetilde{C}  \leq \min_{f \in \mcal{F}_{\mcal{Y}}} f(\mcal{Y}) \stackrel{(a)}{\leq}& \max_{i \in \{1,2\} }R^{(0)}_{Si} + R^{(1)}_{33} + \max_{i \in \{1,2\} }R^{(0)}_{iD} \\
                                                     \stackrel{(b)}{<}&\ \frac{1}{2}\widetilde{C} + \frac{1}{2} \widetilde{C} = \widetilde{C}
                \end{aligned}
            \]
            where $(a)$ follows from row 4 in Table \ref{tab:class_to_f_matching} and $(b)$ follows by applying \eqref{eq:case_2_2_t1_conc} and \eqref{eq:case_2_2_assump}.
            On the other hand, having three or more networks with Type III critical cuts implies (by the Pigeon-Hole Principle) that two of them are of the form $\mcal{W}_1$ or else two are of the form $\mcal{W}_2$.
            Without loss of generality, assume it is $\mcal{W}_1$.
            This implies that either:
            \begin{equation}
                \begin{aligned}
                    &\max_{i \in \{1,2\} }R^{(0)}_{Si} + R^{(1)}_{33} &< \frac{1}{2}\widetilde{C}\\
                    &\text{\textbf{or}}& \\
                    &\max_{i \in \{1,2,3\}} R^{(1)}_{3i}  &< \frac{1}{2}\widetilde{C}
                \end{aligned}
                \label{eq:case_2_2_t2_conc}
            \end{equation}
            Now consider the full network cuts $\mcal{Y}_1 = \{S, \{3\},\{1,2\}\}$ and $\mcal{Y}_2 = \{S,\{3\},\phi\}$. One of these cuts will give a contradiction as follows:
            \[
                \begin{aligned}
                    \widetilde{C}  \leq \min_{f \in \mcal{F}_{\mcal{Y}}} f(\mcal{Y}) \stackrel{(a)}{\leq}& \max_{i \in \{1,2\} }R^{(0)}_{Si} + R^{(1)}_{33} + \max_{i \in \{1,2\} }R^{(0)}_{iD} \\
                    \stackrel{(b)}{<}&\ \frac{1}{2}\widetilde{C} + \frac{1}{2} \widetilde{C} = \widetilde{C}\\
                    &\text{\textbf{or}}& \\
                    \widetilde{C}  \leq \min_{f \in \mcal{F}_{\mcal{Y}}} f(\mcal{Y}) \stackrel{(c)}{\leq}& \max_{i \in \{1,2\} }R^{(0)}_{Si} + \max_{i \in \{1,2,3\}} R^{(1)}_{3i} \\
                                                     <&\ \frac{1}{2}\widetilde{C} + \frac{1}{2} \widetilde{C} = \widetilde{C}
                \end{aligned}
            \]
            where $(a),(c)$ follow from rows 4 and 2 respectively, in Table \ref{tab:class_to_f_matching} while $(b),(d)$ follow from \eqref{eq:case_2_2_assump} and \eqref{eq:case_2_2_t2_conc}.
            If less than two subnetworks have Type II critical cuts and less than three subnetworks have Type III critical cuts, then at least one subnetwork has a critical cut of Type I, for which we showed a contradiction earlier.
            This concludes the contradiction for the case $z_0 =2 , z_2 =2$.

            \subsection{\fbox{$ z_0 = 3$ or $z_2 = 3$}}
            Assume $z_0 = 3$. In this case, we have:
                    \[
                        R^{(0)}_{Si} < \frac{1}{2} \widetilde{C}\quad,\quad i \in \{1,2,3\}
                    \]
                    Therefore, the cut characterized by $\mcal{Y} = \{ S, \phi , \phi \}$ gives a contradiction as follows:
                    \[
                        \begin{aligned}
                            \widetilde{C} \leq \min_{\mcal{Y} \in \mcal{F}_{\mcal{Y}}} f(\mcal{Y}) \leq& \max_{i \in \{1,2,3\}} R^{(0)}_{Si} \\
                                                             <& \frac{1}{2} \widetilde{C}
                        \end{aligned}
                    \]
                    A similar argument follows for the case $ z_2 = 3 $.

\noindent This concludes the proof of Theorem \ref{thm:k2_out_of_3}.


\begin{appendices}
    \section{Proof of Property \ref{prpty:minimalcuts_N2}} \label{appendix:T_Y_proof}
    Define $s_l \triangleq \vert \mb{y}_l \vert \ \forall\ 1 \leq l \leq L$ (it follows that $\vert \mb{y}_l^c \vert = N - s_l $).
    Since $\mcal{M}_0 = \{S\}$, it follows that $s_0 = \vert \mb{y}_0 \vert = 1$. Similarly, since $\mcal{M}_{L+1} = \{D\}$, and $D \not\in \mcal{Y}$, $\vert \mb{y}_{L+1}^c \vert = 1$
Since \eqref{eq:additive_terms} only depends on the cardinalities of $\mb{y}_l$, we can rewrite it in terms of $s_l$ as:
                \begin{align*}
                    \max_\mcal{Y}\ T(\mcal{Y})\  =\  & \max \sum_{i=0}^{L-1} \min(s_i,N-s_{i+1}) + \min(s_L,1)\\
                                           s.t. &\quad 0 \leq s_i \leq N \quad \forall 1 \leq i \leq L, \\
                                                &\quad s_0 = 1,\ s_i \in \mbb{Z}
                \end{align*}
                Let $t_i = \min(s_i,N-s_{i+1}) \forall i \in \{1,2,\dots L+1\}$,
                $t_0 = \min(1,N-s_l)$ and $t_L = \min(s_L,1)$.
                The problem then becomes
                \begin{align*}
                \label{eq:opt_problem_full}
                    \max_\mcal{Y} &\ T(\mcal{Y})\  =\  \max \sum_{i=0}^{L} t_i \\
                                           s.t. &\quad t_0 \leq 1,\quad t_0 \leq N-s_1, \\
                                           &\quad t_i \leq s_i,\quad t_i \leq N-s_{i+1} \ \ \ \ \  \forall 1 \leq i \leq L-1, \numberthis \\
                                                &\quad t_L \leq 1, \quad t_L \leq s_L, \\
                                                &\quad s_i \in \mbb{Z}, \quad 0 \leq s_i \leq N \quad \forall 1 \leq i \leq L
                \end{align*}
                If we relax the optimization problem in \eqref{eq:opt_problem_full} by removing a subset of the constraints, then we get an upper bound to the maximum $T(\mcal{Y})$. Depending on whether $L$ is even or odd, we perform the relaxation differently.
\subsection{$L$ odd:}
Consider the following relaxed version of \eqref{eq:opt_problem_full}:
                \begin{align}
\label{eq:opt_problem_relaxed}
\nonumber                     & \max \sum_{i=0}^{L} t_i \\
                                           s.t. &\quad t_0 \leq 1, \quad t_L \leq 1, \\
\nonumber                                                &\quad t_i \leq s_i \quad\quad\quad\ \ \ \forall i \in \{2,4,6,\cdots,L-1 \}, \\
\nonumber                                                &\quad t_i \leq N-s_{i+1} \quad \forall i \in \{1,3,5,\cdots,L-2\}
                                            \end{align}
Since each variable $t_i$ is upper bounded by one constraint in \eqref{eq:opt_problem_relaxed} and the objective function is monotonically increasing in $t_i$, the optimal value of the objective function in \eqref{eq:opt_problem_relaxed} is
            \[
                \begin{aligned}
                    \max \sum_{i=0}^{L} t_i &= t_0 + \sum_{i=1}^{\frac{L-1}{2}} t_{2i} + \sum_{i=1}^{\frac{L-1}{2}} t_{2i-1} + t_L \\
                    &= 1 + \sum_{i=1}^{\frac{L-1}{2}} s_{2i} + \sum_{i=1}^{\frac{L-1}{2}} (N-s_{2i}) + 1 \\
                    &= 1 + \sum_{i=1}^{\frac{L-1}{2}} N + 1 = 2 + \frac{(L-1)N}{2}
                \end{aligned}
            \]
            Therefore, for odd $L$, we have:
            \[
                \max_{\mcal{Y}} T(\mcal{Y}) \leq \frac{4 + (L-1)N}{2}
            \]

\subsection{$L$ even:}
Through a different relaxation of \eqref{eq:opt_problem_full}, we get:
                \begin{align*}
                     & \max \sum_{i=0}^{L} t_i \\
                                                &s.t.\\
                                                &\quad t_0 \leq 1, \quad t_L \leq s_{L}, \\
                                                &\quad t_i \leq s_i \quad\quad\quad\ \ \ \forall i \in \{2,4,6,\cdots,L-2 \}, \\
                                                &\quad t_i \leq N-s_{i+1} \quad \forall i \in \{1,3,5,\cdots,L-1\}, \\
                \end{align*}
                Again, since each variable in the objective function is upper bounded by a single constraint and the objective function is monotonically increasing in all the variables, we have:
            \[
                \begin{aligned}
                    \max \sum_{i=0}^{L} t_i &= t_0 + \sum_{i=1}^{\frac{L}{2}} t_{2i} + \sum_{i=1}^{\frac{L}{2}} t_{2i-1} \\
                    &= 1 + \sum_{i=1}^{\frac{L}{2}} s_{2i} + \sum_{i=1}^{\frac{L}{2}} (N-s_{2i}) \\
                    &= 1 + \frac{LN}{2} = \frac{LN + 2}{2}
                \end{aligned}
            \]
            Therefore, for even $L$:
            \[
                \max_{\mcal{Y}} T(\mcal{Y}) \leq \frac{LN + 2}{2}
            \]
            This concludes the proof of Property \ref{prpty:minimalcuts_N2}.
\section{Gap Calculation for cuts in Theorem \ref{thm:k2_out_of_3}} \label{appendix:gaps}
In this appendix, following Theorem \ref{thm:k2_out_of_3}, we consider a network with $L=2$ layers of relays and $N=3$ relays per layer.
Our target is to develop an upper bound for $\widebar{C}$ with a controlled constant gap term.
For a cut $\mcal{Y} = \{S,\mb{y}_1,\mb{y}_2\}$, we define the cut capacity
    \begin{align}
        Cut(\mcal{Y}) \triangleq M(0)_{\{S\}}^{\{\mb{y}_1^c\}} + M(1)_{\{\mb{y}_1\}}^{\{\mb{y}_2^c\}} + M(2)_{\{\mb{y}_2\}}^{\{D\}}
        \label{eq:cut_capacity_appendix}
    \end{align}
    Define $\mcal{Q}$ to be a class of cuts $\mcal{Y}$, whose capacities have the same structure.
    Primarily, a cut class is closed under (i) Reindexing of relays within the same layer and (ii) Reflection.
    To demonstrate (i), consider the cuts $\widehat{\mcal{Y}} = \{S,\{2,3\},\{1\}\}$ and $\widetilde{\mcal{Y}} = \{S,\{1,3\},\{3\}\}$.
    If we observe the capacities of these cuts, it is clear that their structures are the same, but with different parameters.
\begin{align*}
    Cut(\widehat{\mcal{Y}}) = R^{(0)}_{S1} + M(1)_{\{2,3\}}^{\{1,2,3\}} \\
    Cut(\widetilde{\mcal{Y}}) = R^{(0)}_{S2} + M(1)_{\{1,3\}}^{\{1,2,3\}} \\
\end{align*}
For (ii), consider the capacities of the cuts $\breve{\mcal{Y}} = \{S,\{2,3\},\phi\}$ and $\mcal{Y}' = \{S, \{1,2,3\},\{1\}\}$.
\begin{align*}
    Cut(\breve{\mcal{Y}}) = R^{(0)}_{S1} + M(1)_{\{2,3\}}^{\{1,2,3\}} \\
    Cut(\mcal{Y}') = M(1)_{\{1,2,3\}}^{\{2,3\}} + R^{(2)}_{1D}
\end{align*}
Again, the structures are similar (link + $2 \times 3$ MIMO).
A simple observation here is that if $\mcal{Y} = \{S,\mb{y}_1, \mb{y}_2\} \in \mcal{Q}$, then its reflection $\mcal{Y}^r = \{S,\mb{y}_2^c, \mb{y}_1^c\} \in \mcal{Q}$.

In this appendix, we are interested in seven cut classes. For each of these, we develop a collection of upper bounds and deduce the constants incurred in each.
Although, the seven classes do not cover all $2^6$ cuts of our network, we will see that any cut outside these seven classes plays no part in our proof of Theorem \ref{thm:k2_out_of_3}.
Therefore, any cut that does not belong to any of these classes, will be represented by its true cut capacity $Cut(\mcal{Y})$ and hence incur no additional constant terms.

Throughout the forthcoming calculation, an inequality stacked with the symbol ``\eqref{eq:upper_functions_a}'' denotes a step performed using the upper bound in \eqref{eq:upper_functions_a}. Similarly, for a step using the upper bounds in \eqref{eq:upper_functions_a2} or \eqref{eq:upper_functions_b}, the inequality shall be stacked with ``\eqref{eq:upper_functions_b}'' or ``\eqref{eq:upper_functions_a2}'', respectively.
\subsection{\fbox{$\mcal{Q}_1 : \mcal{Y}_1 = \{S,\phi, \phi\}$ , $\mcal{Y}_1^r = \{S,\{1,2,3\},\{1,2,3\}\} : $}}
For this cut, we only need one upper bound
\begin{align*}
    Cut(\mcal{Y}_1) = M(1)_{\{S\}}^{\{1,2,3\}} \stackrel{\eqref{eq:upper_functions_b}}{\leq} \max_{i \in \{1,2,3\}} R^{(0)}_{Si} + \log(3)
\end{align*}
Denoting, the non-constant term by $f_{1}$, we have
\begin{align*}
Cut(\mcal{Y}_1) \leq f_{1} + \log(3)
\end{align*}

\subsection{\fbox{$\mcal{Q}_2 : \mcal{Y}_2 = \{S,\{a\}, \phi\}$ , $\mcal{Y}_2^r = \{S,\{1,2,3\},\{a',b'\}\} : $}}
Without loss of generality, let $a = 3$.
We only need one bound for this cut form
\begin{align*}
    Cut(\mcal{Y}_2) &= M(1)_{\{S\}}^{\{1,2\}} + M(1)_{\{3\}}^{\{1,2,3\}} \\
    &\stackrel{\eqref{eq:upper_functions_b}}{\leq} \max_{i \in \{1,2\}} R^{(0)}_{Si} + \log(2) + \max_{i \in \{1,2,3\}} R^{(1)}_{3i} + \log(3) \\
    & = f_{2} + \log(3) + \log(2)
\end{align*}
where the non-constant terms are represented by $f_2$:
\[
    f_2 = \max_{i \in \{1,2\}} R^{(0)}_{Si} + \max_{i \in \{1,2,3\}} R^{(1)}_{3i}
\]

\subsection{\fbox{$\mcal{Q}_3 : \mcal{Y}_3 = \{S,\{a\}, \{b\} \}$ , $\mcal{Y}_3^r = \{S,\{a',b'\},\{c',d'\}\} : $}}
Without loss of generality, let $a = 3$ and $b = 1$.
For this cut form, we have:
\begin{align*}
    C&ut(\mcal{Y}_3) \\
    & =M(1)_{\{S\}}^{\{1,2\}} + M(1)_{\{3\}}^{\{2,3\}} + R^{(2)}_{1D}\\
    &\stackrel{\eqref{eq:upper_functions_b}}{\leq} \max_{i \in \{1,2\}} R^{(0)}_{Si} + \log(2) + \max_{i \in \{2,3\}} R^{(1)}_{3i} + \log(2) + R^{(2)}_{1D}\\
    & = f_{3} + 2\log(2)
\end{align*}
where
\[
    f_3 = \max_{i \in \{1,2\}} R^{(0)}_{Si} + \max_{i \in \{2,3\}} R^{(1)}_{3i} + R^{(0)}_{1D}
\]

\subsection{\fbox{$\mcal{Q}_4 : \mcal{Y}_4 = \{S,\{a\}, \{b,c\} \} : $}}
Without loss of generality, let $a = 3, b = 1$ and $c = 2$.
Similar to previous types, we have:
\begin{align*}
    C&ut(\mcal{Y}_4) \\
    & =M(1)_{\{S\}}^{\{1,2\}} + R^{(1)}_{33} + M(1)_{\{1,2\}}^{\{D\}}\\
    &\stackrel{\eqref{eq:upper_functions_b}}{\leq} \max_{i \in \{1,2\}} R^{(0)}_{Si} + \log(2) + R^{(1)}_{33} + \max_{i \in \{1,2\}} R^{(2)}_{iD} + \log(2) \\
    & = f_{4} + 2\log(2)
\end{align*}
where
\[
    f_4 = \max_{i \in \{1,2\}} R^{(0)}_{Si} + R^{(1)}_{33} + \max_{i \in \{1,2\}} R^{(2)}_{iD}
\]
\subsection{\fbox{$\mcal{Q}_5 : \mcal{Y}_5 = \{S,\{a,b\},\phi\}$ , $\mcal{Y}_5^r = \{S,\{1,2,3\},\{a'\}\} : $}}
    Without loss of generality, consider $a = 2, b = 3$.
    We will be needing three different kinds of upper bounds for this cut, all of which come from upper bounding the $2 \times 3$ MIMO term.
    Consider the following
    \begin{align*}
        C&ut(\mcal{Y}_5) = R^{(0)}_{S1} + M(1)_{\{2,3\}}^{\{1,2,3\}} \\
        &\stackrel{\eqref{eq:upper_functions_b}}{\leq} R^{(0)}_{S1} + \max_{\substack{\mb{v} \subset \{1,2,3\},\\ \vert \mb{v} \vert = 2} }M(1)_{\{2,3\}}^{\{\mb{v}\}} + \log(3) \\
        & = f_{5,1} + \log(3)
    \end{align*}
    Similarly
    \begin{align*}
        C&ut(\mcal{Y}_5) = R^{(0)}_{S1} + M(1)_{\{2,3\}}^{\{1,2,3\}} \\
        & \stackrel{\eqref{eq:upper_functions_a}}{\leq} R^{(0)}_{S1} + M(1)_{\{2\}}^{\{1,2,3\}} + M(1)_{\{3\}}^{\{1,2,3\}} \\
        &\stackrel{\eqref{eq:upper_functions_b}}{\leq} R^{(0)}_{S1} + \hspace{-0.05in}\max_{i \in \{1,2,3\}}\hspace{-0.05in}R^{(1)}_{2i} + \log(3) + \hspace{-0.05in}\max_{i \in \{1,2,3\}} \hspace{-0.05in}R^{(1)}_{3i} + \log(3) \\
        & = f_{5,2} + 2\log(3)
    \end{align*}
    Alternatively, for any $p \in \{1,2,3\}$, we have
    \begin{align*}
        C&ut(\mcal{Y}_5) = R^{(0)}_{S1} + M(1)_{\{2,3\}}^{\{1,2,3\}} \\
        &\stackrel{\eqref{eq:upper_functions_a2}}{\leq} R^{(0)}_{S1} + M(1)_{\{2,3\}}^{\{p\}} + M(1)_{\{2,3\}}^{\{1,2,3\}\backslash\{p\}} \\
        &\stackrel{\eqref{eq:upper_functions_b}}{\leq} R^{(0)}_{S1} + \hspace{-0.05in}\max_{i \in \{2,3\}}\hspace{-0.05in}R^{(1)}_{ip} + \log(2) + M(1)_{\{2,3\}}^{\{1,2,3\}\backslash\{p\}} \\
        & = f_{5,3}(p) + \log(2)
    \end{align*}
    Therefore for full network cuts of the form $\mcal{Y}_5$, we have
    \begin{align*}
        C&ut(\mcal{Y}_5) = R^{(0)}_{S1} + M(1)_{\{2,3\}}^{\{1,2,3\}} \\
        &\leq \min\{f_{5,1} + \log(3), f_{5,2} + 2\log(3), f_{5,3}(p) + \log(2) \} \\
        &\leq \min\{f_{5,1}, f_{5,2}, f_{5,3}(p) \} + 2\log(3) \numberthis
    \end{align*}

\subsection{\fbox{$\mcal{Q}_6 : \mcal{Y}_6 = \{S,\{a,b\},\{c\}\} : $}}
    Without loss of generality, consider $a = 2, b = 3, c = 1$.
    For this cut, we need to compute upper bounds for the $2 \times 2$ MIMO term.
    Consider the following
    \begin{align*}
        C&ut(\mcal{Y}_6) = R^{(0)}_{S1} + M(1)_{\{2,3\}}^{\{2,3\}} + R^{(2)}_{1D} \\
        &\stackrel{\eqref{eq:upper_functions_a}}{\leq} R^{(0)}_{S1} + M(1)_{\{2\}}^{\{2,3\}}+ M(1)_{\{3\}}^{\{2,3\}} + R^{(2)}_{1D} \\
        &\stackrel{\eqref{eq:upper_functions_b}}{\leq} R^{(0)}_{S1} +  \max_{i \in \{2,3\}} R^{(1)}_{2i} + \max_{i \in \{2,3\}} R^{(1)}_{3i} + 2\log(2) + R^{(2)}_{1D} \\
        & = f_{6,1} + 2\log(2)
    \end{align*}
    where
    \[
        f_{6,1} = R^{(0)}_{S1} +  \max_{i \in \{2,3\}} R^{(1)}_{2i} + \max_{i \in \{2,3\}} R^{(1)}_{3i} + R^{(2)}_{1D}
    \]
    Performing the first decomposition a different way, we have:
    \begin{align*}
        C&ut(\mcal{Y}_6) = R^{(0)}_{S1} + M(1)_{\{2,3\}}^{\{2,3\}} + R^{(2)}_{1D} \\
        &\stackrel{\eqref{eq:upper_functions_a2}}{\leq} M(1)^{\{2\}}_{\{2,3\}}+ M(1)^{\{3\}}_{\{2,3\}} \\
        &\stackrel{\eqref{eq:upper_functions_b}}{\leq} R^{(0)}_{S1} + \max_{i \in \{2,3\}} R^{(1)}_{i2} + \max_{i \in \{2,3\}} R^{(1)}_{i3} + 2\log(2) + R^{(2)}_{1D} \\
        & = f_{6,2} + 2\log(2)
    \end{align*}
    where $f_{6,2}$ collects the non-constant terms.
    Finally we can choose not to upper bound the cut capacity at all, i.e.
    \[
        Cut(\mcal{Y}_6) = f_{6,3} = R^{(0)}_{S1} + M(1)_{\{2,3\}}^{\{2,3\}} + R^{(2)}_{1D}
    \]
    Therefore in conclusion, for a full network cut of the form $\mcal{Y}_6$, in general we have
    \begin{align*}
        Cut(\mcal{Y}_6) =&R^{(0)}_{S1} + M(1)_{\{2,3\}}^{\{2,3\}} + R^{(2)}_{1D} \\
        &\leq \min\{f_{6,1} + 2\log(2), f_{6,2} + 2\log(2), f_{6,3} \} \\
        &\leq \min\{f_{6,1}, f_{6,2}, f_{6,3} \} + 2\log(2) \numberthis
    \end{align*}

\subsection{\fbox{$\mcal{Q}_7 : \mcal{Y}_7 = \{S, \{1,2,3\}, \phi\} : $}}
    For this cut, we consider two upper bounds:
    \begin{align}
        \nonumber M(1)_{\{1,2,3\}}^{\{1,2,3\}} &\stackrel{\eqref{eq:upper_functions_b}}{\leq} \frac{3}{2} \max_{ \substack{\mb{u} \subset \{1,2,3\},\\ \mb{v} \subset \{1,2,3\},\\ \vert \mb{u} \vert = \vert \mb{v} \vert = 2}}\hspace{-0.05in} M(1)_{\{\mb{u}\}}^{\{\mb{v}\}} + \frac{3}{2} \log\left(\binom{3}{2}^2\right)  \\
        \nonumber &= \frac{3}{2} \max_{ \substack{\mb{u},\mb{v} \subset \{1,2,3\},\\ \vert \mb{u} \vert = \vert \mb{v} \vert = 2}}\hspace{-0.05in} M(1)_{\{\mb{u}\}}^{\{\mb{v}\}} + 3\log\left(3\right)\\
        &= f_{7,1} + 3\log\left(3\right)
    \end{align}
    We can also upper bound $\mcal{Y}_7$ by applying functions in \eqref{eq:upper_functions} with a different order as:
    \begin{align*}
        M(1)_{\{1,2,3\}}^{\{1,2,3\}} \stackrel{\eqref{eq:upper_functions_a}}{\leq}& M(1)_{\{p\}}^{\{1,2,3\}} + M(1)_{\{1,2,3 \}\backslash\{p\}}^{\{1,2,3\}} \\
        \stackrel{\eqref{eq:upper_functions_b}}{\leq}& \max_{i \in \{1,2,3\}} R^{(1)}_{pi} + \log(3)\\
        &+ \max_{\substack{\mb{v} \subset \{1,2,3\},\\ \vert \mb{v} \vert = 2}}M(1)_{\{1,2,3 \}\backslash\{p\}}^{\{\mb{v}\}} + \log(3) \numberthis
    \end{align*}
    Denoting the non-constant terms as $f_{7,2}(p)$, we have:
    \[
        M(1)_{\{1,2,3\}}^{\{1,2,3\}} \leq f_{7,2}(p) + 2\log(3)
    \]
    Therefore for cuts in $\mcal{Q}_7$, we have
    \begin{align*}
        &Cut(\mcal{Y}_7)\\
        &= M(1)_{\{1,2,3\}}^{\{1,2,3\}} \leq \min\{f_{7,1} + 3\log(3), f_{7,2}(p) + 2\log(3) \} \\
                                     &\leq \min\{f_{7,1}, f_{7,2}(p) \} + 3\log(3)  \numberthis
    \end{align*}
    For any cuts not in the above seven classes, we represent the cut with its true capacity \eqref{eq:cut_capacity_appendix}.
    We summarize the results from this appendix in Table \ref{tab:class_to_f_matching}.

\end{appendices}

\bibliographystyle{IEEEtran}
\bibliography{references}

\end{document}